\documentclass[letterpaper]{article} 
\usepackage[preprint]{aaai2027}  
\usepackage[hyphens]{url}  
\usepackage{graphicx} 
\usepackage{subcaption}
\usepackage{multirow}
\usepackage{natbib}  
\usepackage{caption} 
\usepackage{algorithm}
\usepackage{algorithmic}

\usepackage{amssymb}

\usepackage{newfloat}
\usepackage{amsmath}
\usepackage{amsthm}
\usepackage{listings}

\newtheorem{theorem}{Theorem}
\newtheorem{lemma}{Lemma}
\newtheorem{proposition}{Proposition}
\newtheorem{assumption}{Assumption}

\newcommand{\meanstd}[2]{#1^{\scriptscriptstyle \pm #2}}

\DeclareCaptionStyle{ruled}{labelfont=normalfont,labelsep=colon,strut=off} 
\floatstyle{ruled}
\newfloat{listing}{tb}{lst}{}
\floatname{listing}{Listing}

\usepackage{booktabs}

\title{Latent Two-Sample Testing for Fair Autonomous Vehicle Road Evaluation}
\author {
    Qiujing Lu\textsuperscript{\rm 1}
    Xuanhan Wang\textsuperscript{\rm 1}
    Guanghong Jia\textsuperscript{\rm 1}
     Zhenxia Zhu\textsuperscript{\rm 2}
     Huijuan Lin\textsuperscript{\rm 2}
     Kehua Sheng\textsuperscript{\rm 2}\\
     Zhichao Hou\textsuperscript{\rm 1}
     Shuo Feng\textsuperscript{\rm 1} \corresponding
}
\affiliations {
    \textsuperscript{\rm 1}Tsinghua University\\
    \textsuperscript{\rm 2}BEIJING DIDI INFINITY TECHNOLOGY AND DEVELOPMENT CO., LTD. \\
}

\begin{document}

\maketitle

\begin{abstract}
With the rapid advancement of autonomous vehicle (AV) systems, fast and reliable iteration through road testing has become increasingly critical. However, changes in testing environments make it difficult to disentangle true performance differences between AV versions from extraneous environmental variations, undermining fair and reliable evaluation. This challenge is further compounded by the high-dimensional and unstructured nature of large-scale road testing data, for which effective analysis and comparison methods remain limited. In this work, we address these challenges by introducing a principled framework for distribution-aware AV evaluation. We first learn structured latent representations that map high-dimensional, unstructured road testing data into a compact latent space, enabling effective characterization of scenario distributions. Building on this representation, we propose a two-stage two-sample testing framework that (i) detects and localizes distributional shifts between testing datasets and (ii) calibrates these shifts via importance sampling to reduce evaluation bias and metric estimation error. Experiments on synthetic and real-world road testing data demonstrate the effectiveness of the proposed method.
\end{abstract}

\section{Introduction}
Recent advances in end-to-end learning have significantly accelerated the iteration cycle of autonomous vehicle (AV) systems. New model versions are produced frequently, each incorporating changes in perception, prediction, planning, or control modules. As a consequence, deciding which version should be promoted for further development or deployment places unprecedented demands on the evaluation process. In practice, simulation testing and road testing constitute the two primary evaluation stages prior to autonomous vehicle deployment. Simulation testing scales well, but fidelity remains limited:  replay-based simulation lacks closed-loop interaction once the vehicle under test deviates from recorded trajectories, and generative scenarios do not always guarantee realistic agent behavior. As a result, road testing remains the most reliable means of evaluating AV performance in real operational environments.
\begin{figure}[t]
    \centering
    \includegraphics[width=0.95\linewidth]{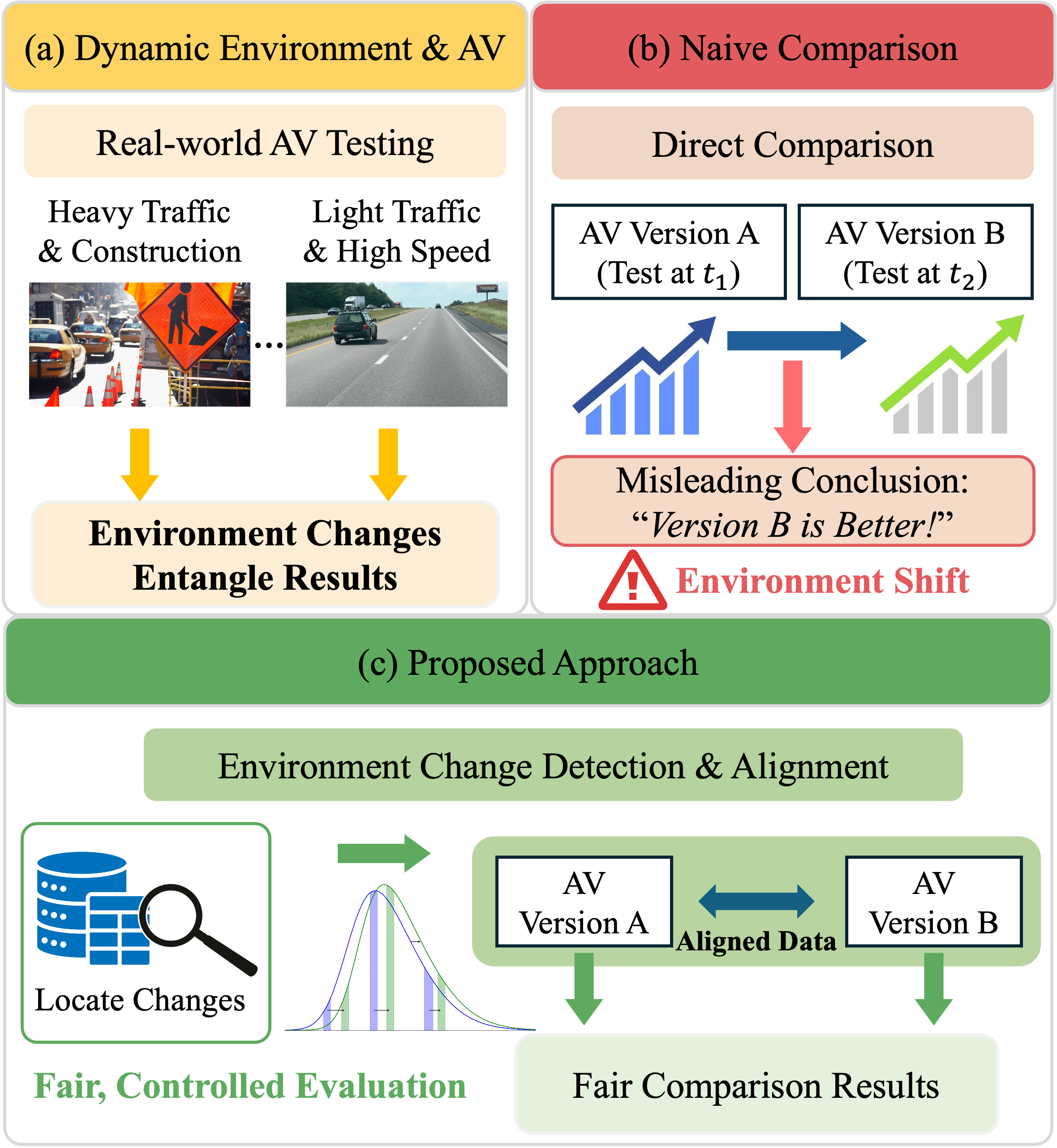}
    \caption{We illustrate evaluation bias caused by distribution shift between environments (P) and (Q). Even with a fixed AV policy, differences in scenario composition (e.g., frequency of cut-in types) lead to biased metric estimates, motivating distribution-aware evaluation and calibration.}
    \label{fig:intro}
\end{figure}
Despite its importance, road testing has long faced a fundamental and practically important challenge: fair AV comparison. Real-world environments are inherently uncontrollable: different AV versions may encounter different scenarios, traffic participants, and environmental contexts, with uneven frequencies that cannot be prescribed in advance. As illustrated in Fig.~\ref{fig:intro}, road-test coverage is therefore incomplete and distribution-dependent. Consequently, observed performance reflects not only the capability of the AV version but also the scenario distribution under which it is tested. Apparent performance gains may thus result from distributional mismatch rather than true capability improvement. This long-standing issue has become increasingly critical as AV systems iterate faster, while limited road-test samples make the resulting testing distributions more dynamic and uneven. However, existing evaluation methods still lack a principled framework for comparing AV versions under distributional differences. Current approaches often rely on aggregated metrics or heuristic scenario grouping, which can mask distributional mismatch and lead to unfair performance comparisons. Moreover, the high dimensionality and unstructured nature of sampled scenarios make direct comparison between road-test datasets difficult.

To address these challenges, we propose a unified framework for large-scale distribution
comparison of autonomous driving data collected from road testing.
Our approach leverages unsupervised context embeddings to map high-dimensional driving logs
into structured latent representations, enabling consistent and interpretable analysis.
Building on this representation, we introduce a two-stage two-sample testing method to localize
and characterize distributional differences, and a distribution calibration strategy that
reduces evaluation bias caused by scenario imbalance, enabling fairer and more reliable performance comparison.

\section{Related work}
\paragraph{Two Sample Testing} The two-sample testing problem seeks to determine whether two independent samples are drawn from the same distribution~\cite{lehmann2005testing}. A wide range of classical methods have been developed for two-sample testing in statistics. Representative examples include the $t$-test for detecting differences in sample means under normality assumptions~\cite{satterthwaite1946approximate,ruxton2006unequal}, as well as nonparametric, distribution-free methods such as the Wilcoxon--Mann--Whitney test~\cite{wilcoxon1945individual,mann1947test}, which is based on rank statistics, and the Kolmogorov--Smirnov test~\cite{an1933sulla,smirnov1939estimation}, which compares empirical cumulative distribution functions. While these approaches are effective in one-dimensional settings, their statistical power and practical applicability deteriorate rapidly as data dimensionality increases. Kernel-based methods such as the Maximum Mean Discrepancy (MMD) test~\cite{gretton2006kernel,gretton2012kernel} extend two-sample testing to high-dimensional settings. Despite strong theoretical guarantees, these approaches depend on manually specified kernels, which fix the underlying function class and limit adaptivity to complex, data-dependent distributional differences. Classifier two-sample tests (C2ST) instead learn a data-adaptive
discriminator and use its held-out classification accuracy as a global discrepancy statistic \cite{lopez2016revisiting}.

\paragraph{Locating Distributional Differences}
While global two-sample tests determine whether two distributions differ, they do not explain how or where these differences arise. This limitation has motivated the study of local two-sample testing, which aims to identify regions or instances responsible for distributional discrepancies. Early work in this area~\cite{duong2013local} leverages kernel density estimation (KDE) to detect local differences between probability density functions. However, density-based approaches scale poorly to high-dimensional spaces and require accurate density estimation, which is notoriously difficult under data scarcity. Kernel mean embedding methods partially address these challenges. The Mean Embedding (ME) test~\cite{jitkrittum2016interpretable} localizes differences through a small set of optimized test locations, offering better interpretability than "black-box" global kernel tests. Nevertheless, ME-based methods provide location-specific probes rather than explicit importance scores for individual data points, and their performance is sensitive to the optimization of these locations. Other approaches explore alternative notions of locality: region-based methods~\cite{soriano2017probabilistic} partition the input space to detect discrepancies in sample cardinality within specific bins, while regression-based methods~\cite{kim2019global} define locality through the lens of classification certainty. The latter identifies differences by estimating the posterior probability of a sample’s group membership at specific coordinates, shifting the focus from comparing local density masses to evaluating the discriminability of samples at a pointwise level.

These methods, however, still struggle with complex, high-dimensional data. Naive partitioning ignores intrinsic geometry, and the reliable estimation of conditional probabilities or densities remains infeasible under finite samples. Despite advances in domain-specific settings—such as astronomy~\cite{freeman2017local} and medical analysis~\cite{zhao2021detection}—developing scalable, data-efficient methods for localizing discrepancies in structured and dynamically evolving data remains an open challenge.

\paragraph{AV evaluation and testing} 
AV testing frameworks have been systematically developed over the past two decades along with multiple international standards and industrial practices \cite{ISO26262,ISO21448, PEGASUS}. However, as AV systems become increasingly complex and safer, globally validating their safety and reliably evaluating performance has become substantially more challenging \cite{huang2016autonomous,liu2024curse,liu2025autonomous}. To address this challenge, scenario-based \cite{cai2022survey} and environment-based testing methodologies \cite{feng2021intelligent,yan2023learning,feng2023dense, wang2024drivedreamer} have been proposed, aiming to ensure sufficient coverage of diverse environmental conditions and interaction patterns. Although simulation-based testing provides scalability and controllability, fidelity limitations in sensor modeling and agent interactions continue to hinder reliable sim-to-real evaluation \cite{lindstrom2024nerfs,birchler2025roadmap}. As Level-4 AV fleets scale, road testing has become increasingly prevalent, yielding large real-world datasets. Despite its cost and limited controllability, road testing remains a uniquely reliable source for evaluation, as it directly reflects real operational environments and reveals distributional variations that are difficult to faithfully reproduce in simulation. This exposes a critical research gap in how to fairly compare road-testing data across deployments under evolving environments.

\begin{figure*}[t]
    \centering
    \includegraphics[width=0.99\textwidth]{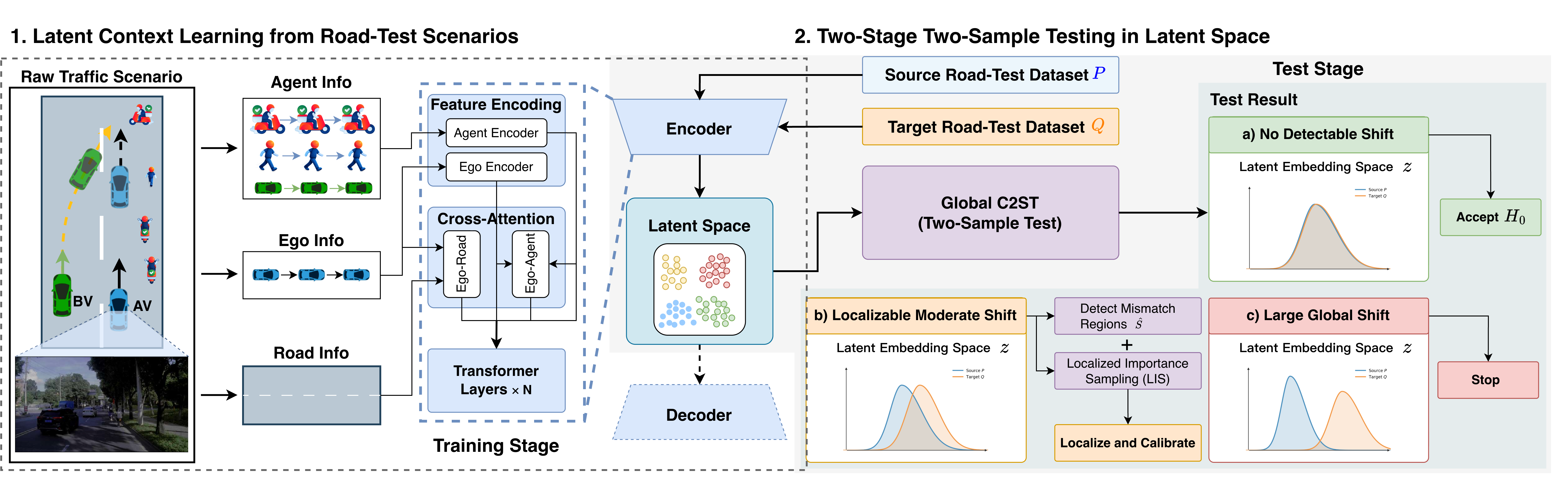}
    \caption{\textbf{Overview of the proposed distribution-aware road-test evaluation framework.} Scenario embeddings are learned from ego, agent, and road features, then compared using global C2ST; moderate shifts are localized and calibrated with LIS, while negligible or large shifts trigger early termination.}
    \label{fig:architecture}
\end{figure*}

\section{Problem Formulation}
\label{sec:problem}

\subsection{Road-Test Scenarios and Performance Evaluation}

We model each road-test scenario as the outcome of a closed-loop interaction
between an AV and its environment. Let $\mathcal{C}$ denote the context space,
including road geometry, infrastructure, traffic composition, and the initial
states of traffic participants. Given a context $C\in\mathcal{C}$ and a fixed
ego policy $\pi_s$, closed-loop execution generates a realized scenario
    $X=\tau(\pi_s;C)\in\mathcal{X}$,
where $\mathcal{X}$ denotes the raw scenario space. A learned encoder $\phi:\mathcal{X}\rightarrow\mathcal{Z}$ maps each scenario to $Z=\phi(X)\in\mathcal{Z}$. Policy performance is measured by a bounded trajectory-level metric
$\ell:\mathcal{X}\rightarrow[0,M]$, such as a safety-event indicator or a
comfort score. The context-conditional performance is
$ L(\pi_s,C)
    =
    \mathbb{E}\!\left[\ell(X)\mid C,\pi_s\right].$

We consider an observed road-testing regime with context distribution $P_C$
and a target regime with distribution $Q_C$. The ego policy is held fixed
across the two regimes, so their difference primarily reflects variation in
scenario exposure. Closed-loop execution induces scenario distributions
$P_X,Q_X$. The target performance is
\begin{equation}
    \mu_Q(\pi_s)
    =
    \mathbb{E}_{C\sim Q_C}[L(\pi_s,C)]
    =
    \mathbb{E}_{X\sim Q_X}[\ell(X)].
    \label{eq:target_metric}
\end{equation}
However, the available road-test samples are drawn from $P_X$, and a naive
evaluation estimates
$ \mu_P(\pi_s) = \mathbb{E}_{X\sim P_X}[\ell(X)]. $
When $P_C\neq Q_C$, the resulting evaluation offset is $\Delta_\ell(P,Q;\pi_s)
    =
    \mu_P(\pi_s)-\mu_Q(\pi_s),$
which may confound environmental variation with AV performance.

\subsection{Quantification of Distribution Shift}
\label{sec:shift_quantification}

Directly comparing $P_X$ and $Q_X$ is challenging because road-test scenarios
are high-dimensional, variable-length multi-agent sequences. We therefore
compare their induced latent distributions $P_Z$ and $Q_Z$.

Consider a balanced binary classification problem with source label
$Y=0$ for $P_Z$ and $Y=1$ for $Q_Z$. Let $a^\star(P_Z,Q_Z)$ denote the
Bayes-optimal classification accuracy and define its advantage over random
guessing as $ \gamma^\star =  a^\star(P_Z,Q_Z)-\frac{1}{2}.$

\begin{proposition}[Classifier advantage and total variation]
\label{prop:classifier_tv}
Under equal class priors,
\begin{equation*}
    \operatorname{TV}(P_Z,Q_Z)
    =
    2\gamma^\star.
    \label{eq:tv_bayes}
\end{equation*}
For any fixed classifier $h$ with population accuracy $a(h)$ and advantage
$\gamma(h)=\max\{0,a(h)-1/2\}$,
\begin{equation*}
    2\gamma(h)
    \leq
    \operatorname{TV}(P_Z,Q_Z).
    \label{eq:classifier_lower_bound}
\end{equation*}
\end{proposition}

Thus, the advantage of a learned classifier measures the discrepancy detectable
by that classifier. The proof and finite-sample confidence bounds for the classifier two-sample test (C2ST) are provided in Appendix~\ref{app:c2st_theory}.

\subsection{Latent Distribution Calibration}
\label{sec:calibration_problem}

Our objective is to estimate the target performance $\mu_Q(\pi_s)$ using
samples drawn from the observed distribution $P_X$. Calibration in latent
space requires two standard conditions.

\begin{assumption}[Latent overlap and conditional invariance]
\label{ass:latent_shift}
The target latent distribution is absolutely continuous with respect to the
observed distribution,
$Q_Z\ll P_Z,$
and the conditional expected metric is invariant across the two regimes:
\begin{equation*}
\begin{aligned}
    \mathbb{E}_{P_X}[\ell(X)\mid Z=z]
    &=
    \mathbb{E}_{Q_X}[\ell(X)\mid Z=z] \\
    &:=m(z),
\end{aligned}
\label{eq:conditional_invariance}
\end{equation*}
for $P_Z$- and $Q_Z$-almost every $z$.
\end{assumption}

The second condition requires the learned representation to preserve the
contextual information relevant to the evaluated metric.

\begin{theorem}[Latent importance-calibration identity]
\label{thm:latent_calibration}
Under Assumption~\ref{ass:latent_shift}, define
\begin{equation*}
    w^\star(z)
    =
    \frac{\mathrm{d}Q_Z}{\mathrm{d}P_Z}(z).
\end{equation*}
Then the target performance satisfies
\begin{equation*}
    \mu_Q(\pi_s)
    =
    \mathbb{E}_{X\sim P_X}
    \left[
        w^\star\!\left(\phi(X)\right)\ell(X)
    \right].
\end{equation*}

Moreover, suppose the distributional shift is confined to a measurable region
$\mathcal{S}\subseteq\mathcal{Z}$ such that
$$
    P_Z(A)=Q_Z(A),
    \qquad
    \text{for every measurable }A\subseteq\mathcal{S}^{c}.
$$
Then $w^\star(z)=1$ for $P_Z$-almost every
$z\in\mathcal{S}^{c}$, and
\begin{equation*}
\begin{aligned}
    \mu_Q(\pi_s)
    =
    \mathbb{E}_{X\sim P_X}\big[
        &\ell(X)\mathbf{1}\{\phi(X)\notin\mathcal{S}\} \\
        &+
        w^\star(\phi(X))
        \ell(X)\mathbf{1}\{\phi(X)\in\mathcal{S}\}
    \big].
\end{aligned}
\end{equation*}
\end{theorem}

Theorem~\ref{thm:latent_calibration} follows directly from the Radon--Nikodym change-of-measure identity and the tower property of conditional expectation \cite{kallenberg2021foundations}. It further
shows that reweighting is required only within latent regions where the two testing distributions differ.

\section{Methodology}
\label{sec:method}

\subsection{Scenario Representation Learning}
\label{sec:ae_impl}

We learn an encoder
$\phi:\mathcal{X}\rightarrow\mathcal{Z}$
that maps each road-test scenario to a compact representation for distribution
comparison. As shown in Fig.~\ref{fig:architecture}, 
the model takes synchronized temporal sequences extracted from real-world road-testing logs as input, including ego-vehicle states, surrounding-agent states, road-structure features, and the corresponding time-validity masks. The architecture follows a Transformer-based encoder--decoder design. 

Agent and ego states are first encoded into a shared feature space, while agent-type and temporal information are incorporated during feature encoding. Two cross-attention branches then model complementary dependencies: the ego--agent branch captures dynamic interactions between the ego vehicle and surrounding traffic participants, whereas the ego--road branch encodes the relationship between ego motion and road context. The resulting interaction-aware and context-aware features are fused with the original ego and agent embeddings at each time step and projected into a unified representation. A temporal Transformer subsequently aggregates these features across the scenario horizon, followed by masked temporal pooling to obtain the scene-level latent representation $Z=\phi(X)$. The decoder reconstructs behavior-relevant scenario attributes, encouraging the latent space to preserve the information required for downstream distribution comparison.

The model is trained using a combination of reconstruction and attribute-prediction objectives:
$\mathcal{L}_{\mathrm{train}}
=
\mathcal{L}_{\mathrm{recon}}
+
\mathcal{L}_{\mathrm{attr}},$
where $\mathcal{L}_{\mathrm{recon}}$ measures the discrepancy between the input scenario and its reconstruction, and $\mathcal{L}_{\mathrm{attr}}$  evaluates the prediction errors of behavior-relevant attributes from the latent representation.

\subsection{Two-Stage Latent Two-Sample Testing}
\label{sec:two_stage_tst}

Let$
    \{X_i^{P}\}_{i=1}^{n_P}\sim P_X,
    \{X_j^{Q}\}_{j=1}^{n_Q}\sim Q_X
$ denote scenarios from the observed and target regimes. Their latent
representations are $
    Z_i^{P}=\phi(X_i^{P})\sim P_Z,
    Z_j^{Q}=\phi(X_j^{Q})\sim Q_Z.$

\paragraph{Stage I: global shift detection.}
We train a probabilistic classifier
$h:\mathcal{Z}\rightarrow[0,1]$
to predict the source label, where
\begin{equation*}
    h(z)
    \approx
    \mathbb{P}(Y=1\mid Z=z)
\end{equation*}
and $Y=1$ denotes the target distribution $Q_Z$. Let $\widehat a$ denote
classification accuracy on an independent balanced test set. We define the
empirical classifier advantage as
\begin{equation}
    \widehat{\gamma}
    =
    \max\left\{0,\widehat a-\frac{1}{2}\right\}.
    \label{eq:empirical_advantage}
\end{equation}
By Proposition~\ref{prop:classifier_tv}, this quantity estimates the distributional discrepancy detectable by the learned classifier.

We use a confidence-adjusted statistic
$\gamma_{\mathrm{LCB}}$ to distinguish three regimes: no statistically
detectable shift, a moderate shift suitable for localization, and a large
global shift for which local analysis provides limited additional information.
Threshold selection is described in Appendix~\ref{appendix:twostage_sample}.

\paragraph{Stage II: local shift detection.}
When localization is triggered, let
$\mathcal{D}_Z=\{Z_i\}_{i=1}^{N}$
denote the pooled latent sample. For each candidate point $Z_i$, let
$\mathcal{N}_k(i)$ denote its $k$ nearest neighbors. We compute the smoothed
classifier posterior
\begin{equation}
    \overline h_i
    =
    \frac{1}{k}
    \sum_{j\in\mathcal{N}_k(i)}
    h(Z_j).
    \label{eq:smoothed_classifier_score}
\end{equation}
Under local alignment, $\overline h_i$ is expected to remain close to
$1/2$. Neighborhoods with statistically significant and sufficiently
supported deviations from $1/2$ are retained as candidate mismatch regions.
Effect-size scoring, support constraints, multiple-testing correction, and
region de-duplication are described in
Appendix~\ref{app:local_shift_detection}.

\begin{figure*}[!t]
    \centering
    \includegraphics[width=0.95\textwidth]
    {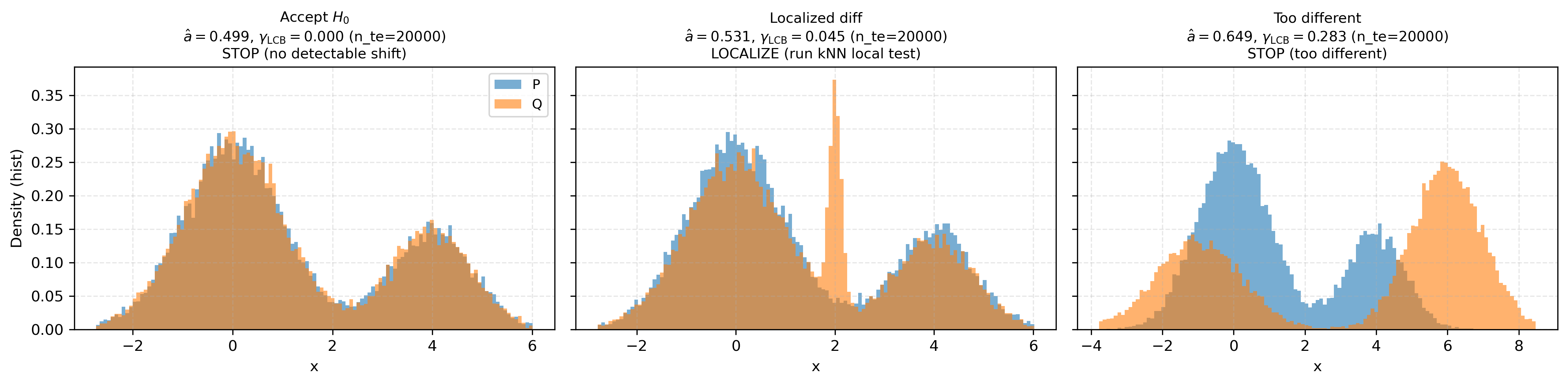}
    \caption{Two-stage two-sample testing under three distribution-shift
    regimes. The method terminates when no shift is detected, localizes
    mismatch regions under moderate shift, and reports a global shift without
    further localization when the distributions are strongly separable.}
    \label{fig:detect_stage}
\end{figure*}

\subsection{Localized Importance Sampling for Metric Calibration}
\label{sec:lis}

Let $ \widehat{\mathcal{S}}
    =
    \bigcup_{k=1}^{K}R_k
    \subseteq\mathcal{Z}$
denote the detected mismatch region. After region de-duplication, overlapping
samples are assigned to the nearest selected region center so that the regions
$\{R_k\}_{k=1}^{K}$ define a unique partition of
$\widehat{\mathcal{S}}$.

Let
\begin{equation*}
    \widehat p_k
    =
    \frac{1}{n_P}
    \sum_{i=1}^{n_P}
    \mathbf{1}\{Z_i^{P}\in R_k\},
    \qquad
    \widehat q_k
    =
    \frac{1}{n_Q}
    \sum_{j=1}^{n_Q}
    \mathbf{1}\{Z_j^{Q}\in R_k\}
\end{equation*}
denote the empirical masses of region $R_k$ under the observed and target
distributions. We estimate a piecewise-constant, clipped density ratio:
\begin{equation}
\widehat w_{\mathrm{loc}}(z)
=
\begin{cases}
\displaystyle
\min\left\{
\frac{\widehat q_k}{\widehat p_k},B
\right\},
& z\in R_k,\\[8pt]
1,
& z\notin\widehat{\mathcal{S}},
\end{cases}
\label{eq:localized_weight}
\end{equation}
where $B>0$ is a clipping threshold. Regions with
$\widehat p_k=0$ are excluded or regularized using a small denominator
constant.

Given observed samples
$X_i^{P}\sim P_X$ and their latent representations
$Z_i^{P}=\phi(X_i^{P})$, we estimate the target performance using the self-normalized estimator
\begin{equation}
    \widehat{\mu}_{Q}^{LIS}
    =
    \frac{
        \sum_{i=1}^{n_P}
        \widehat w_{\mathrm{loc}}(Z_i^{P})
        \ell(X_i^{P})
    }{
        \sum_{i=1}^{n_P}
        \widehat w_{\mathrm{loc}}(Z_i^{P})
    }.
    \label{eq:self_normalized_lis}
\end{equation}
\begin{theorem}[Asymptotic MSE comparison for oracle normalized LIS]
\label{thm:oracle_lis_mse}

Suppose Assumption~\ref{ass:latent_shift} holds and the shift is
confined to a measurable region $S\subseteq\mathcal Z$. Define the
oracle density ratio $ w^\star(z)=\frac{dQ_Z}{dP_Z}(z).$
Since $P_Z$ and $Q_Z$ agree on $S^c$, we have
$w^\star(z)=1$ for $P_Z$-almost every $z\in S^c$, and hence
$
\mathbb E_{P_Z}[w^\star(Z)]=1,
\mathbb E_{P_X}[w^\star(Z)\ell(X)]=\mu_Q.
$

Given $X_1,\ldots,X_n\stackrel{\mathrm{i.i.d.}}{\sim}P_X$ and
$Z_i=\phi(X_i)$, define the oracle normalized LIS estimator
\[
\widehat{\mu}_{Q,\mathrm{or}}^{\mathrm{LIS}}
=
\frac{
    \sum_{i=1}^{n}w^\star(Z_i)\ell(X_i)
}{
    \sum_{i=1}^{n}w^\star(Z_i)
}.
\]

Let $
\sigma_{\mathrm{LIS}}^2
=
\operatorname{Var}_{P_X}
\left[
    w^\star(Z)\bigl(\ell(X)-\mu_Q\bigr)
\right].$
Then
\begin{equation*}
\sqrt n
\left(
    \widehat{\mu}_{Q,\mathrm{or}}^{\mathrm{LIS}}-\mu_Q
\right)
\xrightarrow{d}
\mathcal N(0,\sigma_{\mathrm{LIS}}^2),
\label{eq:oracle_lis_clt}
\end{equation*}
and
\begin{equation*}
\operatorname{MSE}_Q
\left(
    \widehat{\mu}_{Q,\mathrm{or}}^{\mathrm{LIS}}
\right)
=
\frac{\sigma_{\mathrm{LIS}}^2}{n}
+
O\left(\frac{1}{n^2}\right),
\label{eq:oracle_lis_mse}
\end{equation*}
where
$
\operatorname{MSE}_Q(\widehat\mu)
=
\mathbb E\!\left[
    (\widehat\mu-\mu_Q)^2
\right].
$

For the naive estimator
$
\widehat{\mu}_{P}^{\mathrm{naive}}
=
\frac{1}{n}\sum_{i=1}^{n}\ell(X_i),
$
define
$
\mu_P=\mathbb E_{P_X}[\ell(X)],
\Delta=|\mu_P-\mu_Q|,
\sigma_P^2=\operatorname{Var}_{P_X}[\ell(X)].
$
Then
$
\operatorname{MSE}_Q
\left(
    \widehat{\mu}_{P}^{\mathrm{naive}}
\right)
=
\Delta^2+\frac{\sigma_P^2}{n}.
$
Consequently,
\begin{equation*}
\begin{aligned}
&
\operatorname{MSE}_Q
\left(
    \widehat{\mu}_{P}^{\mathrm{naive}}
\right)
-
\operatorname{MSE}_Q
\left(
    \widehat{\mu}_{Q,\mathrm{or}}^{\mathrm{LIS}}
\right)
\\
&\qquad
=
\Delta^2
-
\frac{\sigma_{\mathrm{LIS}}^2-\sigma_P^2}{n}
+
O\left(\frac{1}{n^2}\right).
\end{aligned}
\label{eq:lis_mse_difference}
\end{equation*}

Thus, oracle normalized LIS has a smaller asymptotic MSE whenever
\begin{equation}
\Delta^2
>
\frac{\sigma_{\mathrm{LIS}}^2-\sigma_P^2}{n}
+
O\left(\frac{1}{n^2}\right).
\label{eq:lis_mse_dominance}
\end{equation}
\end{theorem}
The theorem characterizes the bias--variance trade-off of localized calibration. The naive estimator retains the squared shift bias $\Delta^2$, whereas oracle LIS removes this bias while potentially changing the sampling variance. LIS therefore achieves a smaller MSE
when the removed bias exceeds the additional variance induced by importance weighting.

\section{Experiments}
We evaluate the proposed framework in two primary settings. First, we use a synthetic benchmark with controlled and known distributions to validate the proposed two-stage testing and calibration pipeline. Second, we use safety-critical cut-in scenarios as the main real-world benchmark for distribution comparison, localized shift detection, and metric calibration. We further apply the framework to a full-scene road-testing dataset containing multiple scenario types to assess its generalization to more diverse driving contexts.

\subsection{Controlled Validation on Synthetic Data}
We generate three synthetic datasets designed to represent different degrees of distributional shift and apply the proposed two-stage two-sample testing framework. As shown in Fig.~\ref{fig:detect_stage}, these three regimes are correctly identified by our method. We further compare the proposed C2ST with
classical global two-sample testing approaches, including the Kolmogorov--Smirnov test~\cite{an1933sulla,smirnov1939estimation},
kernel-based MMD test~\cite{gretton2006kernel,gretton2012kernel},
and the Wasserstein-1 distance test test~\cite{ramdas2017wasserstein}. 

\begin{table}[h]
\centering
\caption{Global two-sample tests on synthetic data.}
\label{tab:synthetic_twosample}
\setlength{\tabcolsep}{4pt} 
\begin{tabular}{l c c c c c}
\toprule
Case 
& Acc 
& $p_{\text{C2ST}}$ 
& $p_{\text{KS}}$ 
& $p_{\text{MMD}}$ 
& $p_{W_1}$ \\
\midrule
Accept $H_0$
& 0.499
& 0.58
& 0.916
& 0.951
& 0.869 \\
Localized diff
& 0.531
& 0.012
& $5.36{\times}10^{-36}$
& 0.016
& 0.016 \\
Too different
& 0.649
& 0.012
& 0
& 0.016
& 0.016 \\
\bottomrule
\end{tabular}
\end{table}

As shown in Table~\ref{tab:synthetic_twosample}, C2ST provides not only a strong global discrepancy signal through the classification advantage, but also sample-wise probabilistic scores that can be directly reused for localized difference detection. This dual capability makes C2ST particularly well suited to our two-stage, band-pass–gated pipeline, which localizes differences only under moderate distribution shifts. Notably, both the localized-difference and “too-different” cases yield identical p-values under permutation-based C2ST and MMD tests due to saturation at the Monte Carlo resolution limit, highlighting a limitation of relying solely on p-values to characterize distributional shifts. In contrast, the classification advantage offers a continuous, effect-size–aware measure that enables principled gating between localizable and global shifts.

\begin{figure}[!h]
    \centering
    \includegraphics[width=0.9\linewidth]{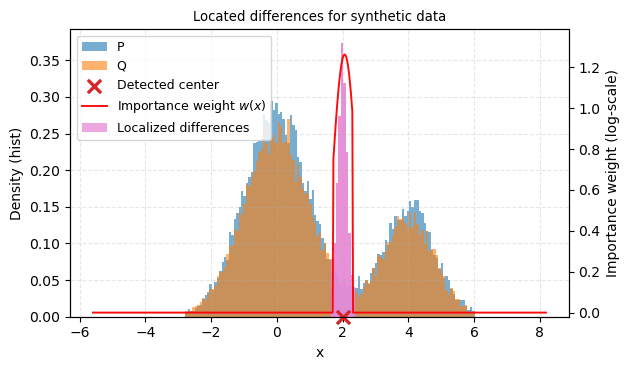}
    \caption{Localized differences and importance weighting on synthetic data. The method identifies a mismatch region (pink) centered at the detected point (red ($\times$)). The estimated importance weight $w(x)$ assigns higher weights to this region, correcting local discrepancies between $P$ and $Q$.}
    \label{fig:syn_local_diff}
\end{figure}
\begin{figure}[!h]
    \centering
    \includegraphics[width=0.98\linewidth,height=7cm]{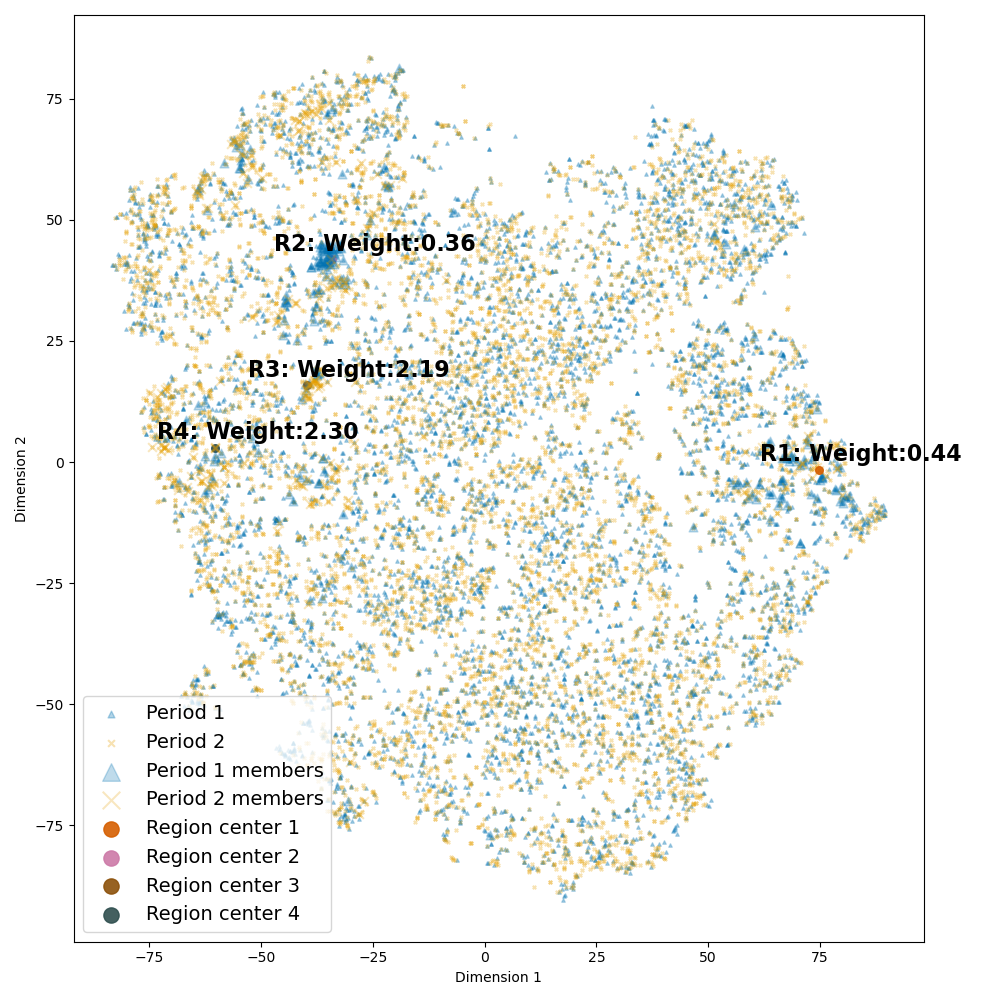}
     \includegraphics[width=0.95\linewidth,height=2cm]{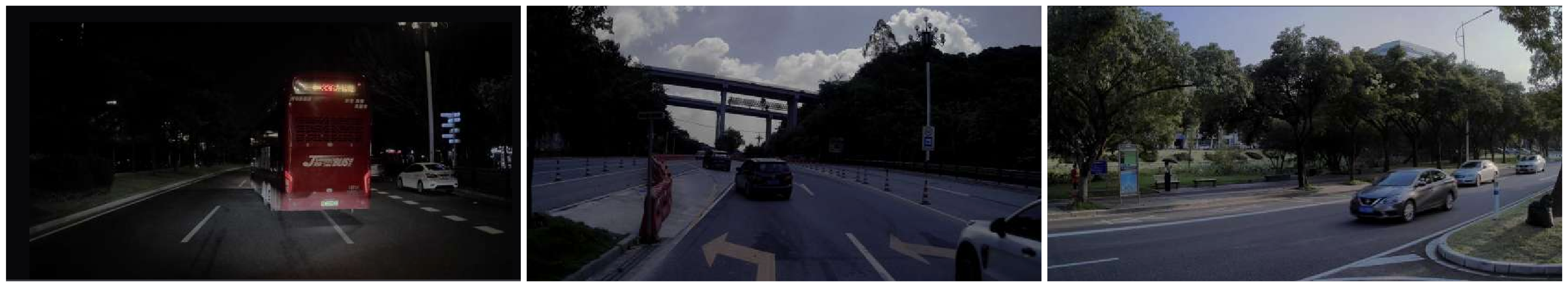}
    \caption{ Localized distribution differences in real-world road-testing data. Top: Latent embeddings from two periods (Period 1 in  blue and Period 2 in orange) with detected mismatch regions (R1–R4). Bottom: representative road-test scenarios sampled from the identified regions. }
    \label{fig:real_local_diff}
\end{figure}
Using the proposed kNN-based localization method in
Section~\ref{sec:two_stage_tst}, we identify the localized distributional shift in the synthetic data, as shown in Fig.~\ref{fig:syn_local_diff}. We define the performance metric as
$\ell(x)=\mathbb{I}(x>1.5)$. Its target expectation under $Q$ is $0.4097$, whereas the naive estimate under $P$ is $0.3466$, yielding an absolute error of $0.0631$. After localized importance sampling, the estimate increases to $0.3915$, reducing the error to $0.0182$, a $71.2\%$ reduction. These results show that LIS effectively corrects the metric bias induced by the localized distribution shift.

\subsection{Real-World Evaluation on Road-Testing Data}
 All data are collected from AV road tests conducted in road-testing mode. For each scenario, we extract synchronized temporal sequences describing the surrounding agents and the ego vehicle from onboard logs. The agent states include position, heading, speed, object type, height, length, and width, while the ego-vehicle states include position, heading, speed, and acceleration. For representation learning, we use one month of road-testing data collected under multiple AV system versions as the training set. 

The evaluation metric is a proprietary in-house measure computed from the ego and surrounding-agent states to approximate driving discomfort. Although the original metric is continuous, we convert it into a binary indicator because our primary interest is whether an uncomfortable-driving event occurs. Specifically, we define $\ell(X)\in\{0,1\},$ where \(\ell(X)=1\) indicates that the discomfort score exceeds a predefined threshold and \(\ell(X)=0\) otherwise. The metric is deployed onboard the test vehicles and evaluated for each recorded scenario.
\paragraph{Structure of the learned embedding space.}
A semantically structured latent space is essential for reliable two-sample testing. We therefore evaluate the learned representation both qualitatively and quantitatively. For the qualitative evaluation, we randomly select anchor scenarios and retrieve their nearest neighbors in the latent space. The retrieved scenarios exhibit similar traffic configurations and interaction patterns, such as low-speed close-range cut-ins, high-speed rear cut-ins, and crossing or wrong-way behaviors. This neighborhood consistency indicates that the encoder organizes scenarios according to their underlying interaction semantics. Additional retrieval examples are provided in Appendix \ref{app:real_exp}.

We further cluster the learned cut-in embeddings into 27 groups and examine how the events of interest are distributed across these clusters. Their cluster-level proportions are clearly non-uniform, with a standard deviation of $3.04$ percentage points.  This heterogeneity suggests that the learned representation distinguishes interaction patterns associated with different levels of safety relevance, rather than producing arbitrary partitions with uniformly distributed events. Moreover, the monthly event distribution is strongly correlated with the independently collected yearly issue distribution, achieving a Pearson correlation of $r=0.758$ ($p=4.57\times10^{-6}$) and Spearman rank correlation  $\rho=0.840$ ($p=4.11\times10^{-8}$). These results indicate substantial agreement in both the relative frequencies and rankings of recurring interaction patterns, providing external evidence that the discovered clusters capture meaningful and temporally consistent safety-related structures.

\paragraph{Global distributional differences.}
We evaluate whether the proposed framework can quantify global distributional shifts in real-world road-testing data. Following Proposition~\ref{prop:classifier_tv}, we train a binary classifier on the learned latent representations to distinguish datasets collected under two settings: different weeks within the same region and different regions within the same week. The autonomous driving system version is fixed in all comparisons to isolate changes in the operational data distribution.

\begin{table}[t]
\centering
\caption{
Global C2ST results under a fixed autonomous driving system version. Here, $\hat{\gamma}$ is the empirical classifier advantage defined in
Eq.~\eqref{eq:empirical_advantage}.
}
\label{tab:real_world_c2st_comparison}
\small
\setlength{\tabcolsep}{6pt}
\begin{tabular}{lcc}
\toprule
Comparison
& ACC
& $2\hat{\gamma}$ \\
\midrule
Cross-week
& $53.3\% \pm 0.9\%$
& $0.067 \pm 0.018$ \\
Cross-region
& $66.5\% \pm 0.6\%$
& $0.330 \pm 0.013$ \\
\bottomrule
\end{tabular}
\end{table}
The two settings exhibit markedly different levels of separability. For cross-week comparisons, the classifier advantage remains small and stabilizes at approximately $\hat{\gamma}=0.030$--$0.035$, corresponding to an average accuracy of $53.4\%$. In contrast, cross-region comparisons yield $\hat{\gamma}=0.150$--$0.165$ and an average accuracy of $66.5\%$. The resulting discrepancies, $2\hat{\gamma}$, are therefore $0.060$--$0.070$ and $0.300$--$0.330$, respectively. These results indicate that temporal shifts within a region are relatively weak, whereas geographic shifts are substantially more pronounced.

This distinction is consistent with the underlying data: regional comparisons capture persistent differences in road structure, traffic composition, and interaction patterns, while week-to-week comparisons primarily reflect shorter-term variation within a shared operational context. The latter can therefore be further examined through localized shift detection and metric calibration. Additional theoretical and experimental analyses are provided in Appendices~\ref{app:lis_theory} and~\ref{appendix:classifier_details}.

\paragraph{Detection of Local Distribution Shifts}
Leveraging the learned latent embeddings, we detect and visualize distributional differences in real-world road-testing data within the latent space, as shown in Fig.~\ref{fig:real_local_diff}. Latent codes from week-to-week data are first reduced to 30 dimensions via PCA and then projected into two dimensions using UMAP \cite{mcinnes2018umap}. The resulting embeddings are colored by period (blue and orange), with localized distributional differences highlighted. The analysis identifies several interpretable sources of distribution shift, including changes in testing routes, such as a higher frequency of U-turns; changes in the dynamic interaction environment, such as an increased occurrence of cut-ins by large vehicles; and changes in the static environment, such as temporary road modifications that induce consistent ego-vehicle deceleration. Additional details is provided in Appendix~\ref{app:vis_latent}.


\paragraph{Metric calibration via LIS}
To mitigate evaluation bias induced by distribution shift, we apply LIS to calibrate metric estimates from the source distribution $P$ to the target distribution $Q$. In this experiment, we use road-testing data collected over three separate weeks. We treat the first period, $P_1$, as the source distribution and calibrate it to subsequent periods, $Q_1$ and $Q_2$. As shown in Fig.~\ref{fig:calibration_comparison}, the LIS-calibrated estimates move consistently toward the
corresponding target ground-truth values. For $Q_1$ and $Q_2$, the remaining relative errors are only $0.96\%$ and $0.50\%$, respectively, demonstrating that LIS effectively corrects the bias induced by temporal changes in the road-testing distribution. 

Additional results on a full-scene road-testing dataset with multiple scenario categories are provided in  Appendix \ref{app:real_exp}.

\begin{figure}[t]
    \centering
    \includegraphics[width=0.45\linewidth,height=3.2cm]{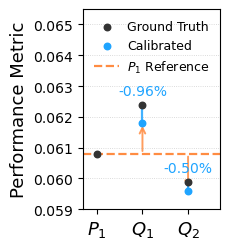}
    \caption{Cross-period metric calibration on real-world road-testing data. LIS adjusts the $P_1$ estimate toward the ground-truth metric values under $Q_1$ and $Q_2$.}
    \label{fig:calibration_comparison}
\end{figure}

\subsection{Ablation Studies}
We evaluate latent representations using three downstream metrics. ACC is the held-out accuracy of the global classifier two-sample test. Jaccard measures overlap between mismatch regions detected by each variant and the base model. LIS-L1 is the absolute error between the LIS-calibrated estimate and the target ground truth, reported in units of $10^{-3}$. For variants that detect no local mismatch regions, Jaccard is reported as zero and LIS-L1 as N/A.

Our default representation uses a 64D latent space. To assess latent capacity, we also train 10D and 128D variants, together with a 64D model using KL regularization. Global detection remains stable across all settings (ACC $\approx$ 0.54--0.56, reject = 100\%), indicating robustness to representation choice. In contrast, local consistency varies non-monotonically with latent capacity: the 10D representation produces fragmented regions with near-zero Jaccard overlap, whereas the 128D and KL-regularized variants detect no mismatch regions, yielding Jaccard = 0 and undefined LIS-L1. These results suggest that insufficient capacity loses interaction-relevant information, while excessive capacity or regularization weakens local density contrast. Overall, 64D provides the best balance between global detection and preservation of local structure. Additional classifier ablations are provided in Appendix \ref{app:real_exp}.

\begin{table}[t]
\centering
\caption{Effect of the latent representation on global detection, local consistency, and calibration (mean $\pm$ standard deviation over three random seeds).}
\label{tab:latent_ablation}
\small
\setlength{\tabcolsep}{3pt}
\begin{tabular}{lccc}
\toprule
Latent setting
& ACC $\uparrow$
& Jaccard $\uparrow$
& LIS-L1 ($\times 10^{-3}$) $\downarrow$ \\
\midrule
latent10
& $\meanstd{0.558}{0.002}$
& $\meanstd{0.00335}{0.001}$
& $\meanstd{1.485}{0.023}$ \\

latent64 (ours)
& $\meanstd{0.547}{0.001}$
& $\meanstd{1.000}{0.000}$
& $\meanstd{1.450}{0.020}$ \\

latent64 + KL
& $\meanstd{0.539}{0.002}$
& $\meanstd{0.000}{0.000}$
& N/A \\

latent128
& $\meanstd{0.541}{0.001}$
& $\meanstd{0.000}{0.000}$
& N/A \\
\bottomrule
\end{tabular}
\end{table}

\section{Conclusion and Future work}
This work takes a first step toward principled, distribution-aware evaluation of
large-scale autonomous driving systems by localizing scenario shifts and calibrating
performance metrics from real-world road testing data. Our results demonstrate that the proposed framework can identify interpretable distributional differences and reduce context-induced evaluation bias. Looking ahead, we aim to extend this framework across the full AV evaluation stack, unifying simulation, closed-loop testing, and real-world deployment
under a common distributional perspective.
In particular, mapping and calibrating distributions between simulated and real-world
domains would enable systematic quantification of the sim-to-real gap, support
risk-aware scenario transfer, and ultimately facilitate scalable, trustworthy end-to-end validation for safety-critical autonomous driving systems.

\newpage
\bibliography{aaai2027}

@book{lehmann2005testing,
  title={Testing statistical hypotheses},
  author={Lehmann, Erich Leo and Romano, Joseph P},
  year={2005},
  publisher={Springer}
}

@article{lopez2016revisiting,
  title={Revisiting classifier two-sample tests},
  author={Lopez-Paz, David and Oquab, Maxime},
  journal={arXiv preprint arXiv:1610.06545},
  year={2016}
}

@article{kim2019global,
  title={Global and local two-sample tests via regression},
  author={Kim, Ilmun and Lee, Ann B and Lei, Jing},
  year={2019}
}

@article{mcinnes2018umap,
  title={Umap: Uniform manifold approximation and projection for dimension reduction},
  author={McInnes, Leland and Healy, John and Melville, James},
  journal={arXiv preprint arXiv:1802.03426},
  year={2018}
}

@article{gretton2012kernel,
  title={A kernel two-sample test},
  author={Gretton, Arthur and Borgwardt, Karsten M and Rasch, Malte J and Sch{\"o}lkopf, Bernhard and Smola, Alexander},
  journal={The journal of machine learning research},
  volume={13},
  number={1},
  pages={723--773},
  year={2012},
  publisher={JMLR. org}
}

@article{jitkrittum2016interpretable,
  title={Interpretable distribution features with maximum testing power},
  author={Jitkrittum, Wittawat and Szab{\'o}, Zolt{\'a}n and Chwialkowski, Kacper P and Gretton, Arthur},
  journal={Advances in Neural Information Processing Systems},
  volume={29},
  year={2016}
}

@article{satterthwaite1946approximate,
  title={An approximate distribution of estimates of variance components},
  author={Satterthwaite, Franklin E},
  journal={Biometrics bulletin},
  volume={2},
  number={6},
  pages={110--114},
  year={1946},
  publisher={JSTOR}
}

@article{ruxton2006unequal,
  title={The unequal variance t-test is an underused alternative to Student's t-test and the Mann--Whitney U test},
  author={Ruxton, Graeme D},
  journal={Behavioral Ecology},
  volume={17},
  number={4},
  pages={688--690},
  year={2006},
  publisher={Oxford University Press}
}

@article{wilcoxon1945individual,
  title={Individual comparisons by ranking methods},
  author={Wilcoxon, Frank},
  journal={Biometrics bulletin},
  volume={1},
  number={6},
  pages={80--83},
  year={1945},
  publisher={JSTOR}
}

@article{mann1947test,
  title={On a test of whether one of two random variables is stochastically larger than the other},
  author={Mann, Henry B and Whitney, Donald R},
  journal={The annals of mathematical statistics},
  pages={50--60},
  year={1947},
  publisher={JSTOR}
}

@article{ramdas2017wasserstein,
  title={On Wasserstein Two-Sample Testing and Related Families of Nonparametric Tests},
  author={Ramdas, Aaditya and Trillos, Nicol{\'a}s and Cuturi, Marco},
  journal={Entropy},
  volume={19},
  number={2},
  pages={47},
  year={2017},
  publisher={MDPI AG}
}

@article{an1933sulla,
  title={Sulla determinazione empirica di una legge didi stribuzione},
  author={An, Kolmogorov},
  journal={Giorn Dell'inst Ital Degli Att},
  volume={4},
  pages={89--91},
  year={1933}
}

@article{smirnov1939estimation,
  title={On the estimation of the discrepancy between empirical curves of distribution for two independent samples},
  author={Smirnov, Nikolai V},
  journal={Bull. Math. Univ. Moscou},
  volume={2},
  number={2},
  pages={3--14},
  year={1939}
}

@article{gretton2006kernel,
  title={A kernel method for the two-sample-problem},
  author={Gretton, Arthur and Borgwardt, Karsten and Rasch, Malte and Sch{\"o}lkopf, Bernhard and Smola, Alex},
  journal={Advances in neural information processing systems},
  volume={19},
  year={2006}
}

@article{duong2013local,
  title={Local significant differences from nonparametric two-sample tests},
  author={Duong, Tarn},
  journal={Journal of Nonparametric Statistics},
  volume={25},
  number={3},
  pages={635--645},
  year={2013},
  publisher={Taylor \& Francis}
}

@article{soriano2017probabilistic,
  title={Probabilistic multi-resolution scanning for two-sample differences},
  author={Soriano, Jacopo and Ma, Li},
  journal={Journal of the Royal Statistical Society Series B: Statistical Methodology},
  volume={79},
  number={2},
  pages={547--572},
  year={2017},
  publisher={Oxford University Press}
}

@article{freeman2017local,
  title={Local two-sample testing: a new tool for analysing high-dimensional astronomical data},
  author={Freeman, PE and Kim, I and Lee, AB},
  journal={Monthly Notices of the Royal Astronomical Society},
  volume={471},
  number={3},
  pages={3273--3282},
  year={2017},
  publisher={Oxford University Press}
}

@standard{ISO26262,
  title        = {{ISO 26262:2018} Road vehicles --- Functional safety},
  organization = {International Organization for Standardization},
  year         = {2018},
  address      = {Geneva, Switzerland}
}

@standard{ISO21448,
  title        = {{ISO 21448:2022} Road vehicles --- Safety of the intended functionality (SOTIF)},
  organization = {International Organization for Standardization},
  year         = {2022},
  address      = {Geneva, Switzerland}
}

@techreport{PEGASUS,
  title        = {PEGASUS Method for the Verification and Validation of Highly Automated Driving Systems},
  author       = {{PEGASUS Project Consortium}},
  institution  = {German Federal Ministry for Economic Affairs and Energy (BMWi)},
  year         = {2019}
}

@article{yan2023learning,
  title={Learning naturalistic driving environment with statistical realism},
  author={Yan, Xintao and Zou, Zhengxia and Feng, Shuo and Zhu, Haojie and Sun, Haowei and Liu, Henry X},
  journal={Nature communications},
  volume={14},
  number={1},
  pages={2037},
  year={2023},
  publisher={Nature Publishing Group UK London}
}

@inproceedings{lindstrom2024nerfs,
  title={Are nerfs ready for autonomous driving? towards closing the real-to-simulation gap},
  author={Lindstr{\"o}m, Carl and Hess, Georg and Lilja, Adam and Fatemi, Maryam and Hammarstrand, Lars and Petersson, Christoffer and Svensson, Lennart},
  booktitle={Proceedings of the IEEE/CVF Conference on Computer Vision and Pattern Recognition},
  pages={4461--4471},
  year={2024}
}

@article{liu2024curse,
  title={Curse of rarity for autonomous vehicles},
  author={Liu, Henry X and Feng, Shuo},
  journal={nature communications},
  volume={15},
  number={1},
  pages={4808},
  year={2024},
  publisher={Nature Publishing Group UK London}
}

@inproceedings{huang2016autonomous,
  title={Autonomous vehicles testing methods review},
  author={Huang, WuLing and Wang, Kunfeng and Lv, Yisheng and Zhu, FengHua},
  booktitle={2016 IEEE 19th International Conference on Intelligent Transportation Systems (ITSC)},
  pages={163--168},
  year={2016},
  organization={IEEE}
}

@article{birchler2025roadmap,
  title={A roadmap for simulation-based testing of autonomous cyber-physical systems: Challenges and future direction},
  author={Birchler, Christian and Khatiri, Sajad and Rani, Pooja and Kehrer, Timo and Panichella, Sebastiano},
  journal={ACM Transactions on Software Engineering and Methodology},
  volume={34},
  number={5},
  pages={1--9},
  year={2025},
  publisher={ACM New York, NY}
}

@article{cai2022survey,
  title={A survey on data-driven scenario generation for automated vehicle testing},
  author={Cai, Jinkang and Deng, Weiwen and Guang, Haoran and Wang, Ying and Li, Jiangkun and Ding, Juan},
  journal={Machines},
  volume={10},
  number={11},
  pages={1101},
  year={2022},
  publisher={MDPI}
}

@inproceedings{wang2024drivedreamer,
  title={Drivedreamer: Towards real-world-drive world models for autonomous driving},
  author={Wang, Xiaofeng and Zhu, Zheng and Huang, Guan and Chen, Xinze and Zhu, Jiagang and Lu, Jiwen},
  booktitle={European conference on computer vision},
  pages={55--72},
  year={2024},
  organization={Springer}
}

@article{feng2021intelligent,
  title={Intelligent driving intelligence test for autonomous vehicles with naturalistic and adversarial environment},
  author={Feng, Shuo and Yan, Xintao and Sun, Haowei and Feng, Yiheng and Liu, Henry X},
  journal={Nature communications},
  volume={12},
  number={1},
  pages={748},
  year={2021},
  publisher={Nature Publishing Group UK London}
}

@book{kallenberg2021foundations,
  author    = {Kallenberg, Olav},
  title     = {Foundations of Modern Probability},
  edition   = {3},
  year      = {2021},
  publisher = {Springer},
  address   = {Cham},
  series    = {Probability Theory and Stochastic Modelling},
  doi       = {10.1007/978-3-030-61871-1}
}

@article{zhao2021detection,
  title={Detection of differentially abundant cell subpopulations in scRNA-seq data},
  author={Zhao, Jun and Jaffe, Ariel and Li, Henry and Lindenbaum, Ofir and Sefik, Esen and Jackson, Ruaidhr{\'\i} and Cheng, Xiuyuan and Flavell, Richard A and Kluger, Yuval},
  journal={Proceedings of the National Academy of Sciences},
  volume={118},
  number={22},
  pages={e2100293118},
  year={2021},
  publisher={National Academy of Sciences}
}

@article{liu2025autonomous,
  title={Autonomous vehicles: A critical review (2004-2024) and a vision for the future},
  author={Liu, Henry and Cao, Zhong and Yan, Xintao and Feng, Shuo and Lu, Qiujing},
  journal={Authorea Preprints},
  year={2025},
  publisher={Authorea}
}

@article{feng2023dense,
  title={Dense reinforcement learning for safety validation of autonomous vehicles},
  author={Feng, Shuo and Sun, Haowei and Yan, Xintao and Zhu, Haojie and Zou, Zhengxia and Shen, Shengyin and Liu, Henry X},
  journal={Nature},
  volume={615},
  number={7953},
  pages={620--627},
  year={2023},
  publisher={Nature Publishing Group UK London}
}
\newpage
\appendix
\section{Theoretical Investigation}
\subsection{Classifier Two-Sample Testing}
\label{app:c2st_theory}

\paragraph{Regression view of C2ST.}
Let $P_Z$ and $Q_Z$ be probability distributions on the measurable
latent space $(\mathcal Z,\mathcal A)$, with densities $p_Z$ and $q_Z$
with respect to a common dominating measure $\lambda$. Consider the
two-sample hypotheses
\[
H_0:P_Z=Q_Z,
\qquad
H_1:P_Z\neq Q_Z.
\]
Introduce a distribution label $Y\in\{0,1\}$ with
\[
\mathbb P(Y=0)=\pi_0,
\qquad
\mathbb P(Y=1)=\pi_1,
\qquad
\pi_0+\pi_1=1,
\]
and let
\[
Z\mid Y=0\sim P_Z,
\qquad
Z\mid Y=1\sim Q_Z.
\]

\begin{lemma}[Regression view of two-sample testing]
\label{lemma:c2st_regression}
The posterior class probability
\[
\eta(z)
:=
\mathbb P(Y=1\mid Z=z)
\]
satisfies
\begin{equation}
\eta(z)
=
\frac{\pi_1q_Z(z)}
{\pi_0p_Z(z)+\pi_1q_Z(z)}
\qquad
\lambda\text{-a.e. }z\in\mathcal Z.
\label{eq:c2st_posterior}
\end{equation}
Moreover, the following statements are equivalent:
\[
P_Z=Q_Z,
\qquad
Y\perp Z,
\qquad
\eta(z)=\pi_1
\quad\lambda\text{-a.e.}
\]
In particular, under equal class priors,
\[
P_Z=Q_Z
\quad\Longleftrightarrow\quad
\eta(z)=\frac12
\quad\lambda\text{-a.e.}
\]
\end{lemma}

\begin{proof}
Equation~\eqref{eq:c2st_posterior} follows directly from Bayes'
rule. If $P_Z=Q_Z$, then $p_Z=q_Z$ almost everywhere and hence
$\eta(z)=\pi_1$. Conversely, if $\eta(z)=\pi_1$ almost everywhere,
then
\[
\pi_1q_Z(z)
=
\pi_1\bigl(\pi_0p_Z(z)+\pi_1q_Z(z)\bigr),
\]
which, since $\pi_0,\pi_1>0$, implies
$p_Z(z)=q_Z(z)$ almost everywhere.
\end{proof}

In practice, the learned probabilistic classifier
$h:\mathcal Z\rightarrow[0,1]$ estimates this posterior:
\[
h(z)\approx\eta(z).
\]
Its induced hard decision rule is
\[
g_h(z)
=
\mathbb{I}\!\left\{h(z)\geq\frac12\right\}.
\]

Lemma~\ref{lemma:c2st_regression} provides the regression formulation of
C2ST, while Proposition~\ref{prop:classifier_tv} quantifies the largest
classification advantage attainable under this formulation.

\paragraph{Classifier advantage and total variation.}
Under equal class priors, define the population accuracy of a decision
rule $g:\mathcal Z\rightarrow\{0,1\}$ as
\[
a(g)
=
\frac12P_Z\{g(Z)=0\}
+
\frac12Q_Z\{g(Z)=1\},
\]
and its advantage over random guessing as
\[
\gamma(g)
=
\max\left\{0,a(g)-\frac12\right\}.
\]

\begin{proof}[Proof of Proposition~\ref{prop:classifier_tv}]
Let
\[
A_g=\{z\in\mathcal Z:g(z)=0\}
\]
denote the region assigned to the source distribution. Then
\begin{align}
a(g)
&=
\frac12P_Z(A_g)
+
\frac12Q_Z(A_g^c) \notag\\
&=
\frac12
\left[
1+P_Z(A_g)-Q_Z(A_g)
\right].
\label{eq:c2st_accuracy_set}
\end{align}
Maximizing over all measurable decision regions gives
\[
a^\star(P_Z,Q_Z)
=
\frac12
\left[
1+
\sup_{A\in\mathcal A}
\{P_Z(A)-Q_Z(A)\}
\right].
\]
Because taking complements reverses the sign of
$P_Z(A)-Q_Z(A)$,
\[
\sup_{A\in\mathcal A}
\{P_Z(A)-Q_Z(A)\}
=
\sup_{A\in\mathcal A}
|P_Z(A)-Q_Z(A)|
=
\operatorname{TV}(P_Z,Q_Z).
\]
Therefore,
\[
a^\star(P_Z,Q_Z)
=
\frac12
\left[
1+\operatorname{TV}(P_Z,Q_Z)
\right],
\]
and hence
\[
2\gamma^\star
=
\operatorname{TV}(P_Z,Q_Z).
\]
For any fixed decision rule $g$,
\[
2\gamma(g)
\leq
\operatorname{TV}(P_Z,Q_Z),
\]
because its decision region is one candidate in the supremum defining
total variation.
\end{proof}

\paragraph{Finite-sample uncertainty.}
Let $g$ be a fixed classifier trained independently of a balanced
held-out test set
\[
\{(Z_i,Y_i)\}_{i=1}^{N_{\mathrm{te}}}.
\]
Define its empirical test accuracy as
\[
\widehat a
=
\frac{1}{N_{\mathrm{te}}}
\sum_{i=1}^{N_{\mathrm{te}}}
\mathbb{I}\{g(Z_i)=Y_i\},
\]
and define
\[
\widehat\gamma
=
\max\left\{0,\widehat a-\frac12\right\}.
\]

Conditional on the trained classifier $g$,
\[
\mathbb E[\widehat a\mid g]=a(g),
\qquad
\operatorname{Var}(\widehat a\mid g)
\leq
\frac{1}{4N_{\mathrm{te}}}.
\]
Moreover, Hoeffding's inequality gives, for any
$\delta\in(0,1)$,
\begin{equation}
\mathbb P\left(
|\widehat a-a(g)|
>
\sqrt{
\frac{\log(2/\delta)}
{2N_{\mathrm{te}}}
}
\,\middle|\,g
\right)
\leq\delta.
\label{eq:c2st_accuracy_bound}
\end{equation}
Because the map
$x\mapsto\max\{0,x-1/2\}$ is $1$-Lipschitz, the same bound implies
\begin{equation}
\left|
\widehat\gamma-\gamma(g)
\right|
\leq
\sqrt{
\frac{\log(2/\delta)}
{2N_{\mathrm{te}}}
}
\label{eq:c2st_advantage_bound}
\end{equation}
with probability at least $1-\delta$.

A distribution-free one-sided lower confidence bound is therefore
\[
\gamma_{\mathrm{LCB}}^{\mathrm{H}}
=
\max\left\{
0,
\widehat a-\frac12
-
\sqrt{
\frac{\log(1/\delta)}
{2N_{\mathrm{te}}}
}
\right\}.
\]
In implementation, we use the tighter Wilson lower confidence bound
defined in Appendix~\ref{app:bandpass_thresholds}.

\paragraph{Training and test-set variability.}
Let $G$ denote the random classifier produced by the training
procedure, including the effects of training-data sampling,
initialization, and optimization. The law of total variance gives
\begin{equation}
\operatorname{Var}(\widehat a)
=
\mathbb E_G
\left[
\operatorname{Var}(\widehat a\mid G)
\right]
+
\operatorname{Var}_G
\left[
a(G)
\right].
\label{eq:c2st_variance_decomposition}
\end{equation}
The first term represents finite test-set uncertainty and is bounded
by $1/(4N_{\mathrm{te}})$, whereas the second term captures
training-induced variability. We therefore report the mean and
standard deviation of $\widehat\gamma$ over repeated classifier
training runs, while separately analyzing test-set uncertainty through
Eq.~\eqref{eq:c2st_accuracy_bound}.

\subsection{Band-Pass Threshold Selection}
\label{app:bandpass_thresholds}
\subsubsection{Selection of Band-Pass Thresholds for Two-Stage Two-Sample Testing}
\label{appendix:twostage_sample}

This section explains how the lower and upper thresholds,
$\gamma_{\mathrm{small}}$ and $\gamma_{\mathrm{large}}$,
are chosen for the band-pass decision rule that determines whether localized
two-sample testing is performed.

Recall that we define the classifier-based shift magnitude proxy as
\begin{equation}
\widehat{\gamma}
=
\max\left\{
0,\widehat a-\frac12
\right\}.
\end{equation}
where $\hat a$ is the held-out classification accuracy of C2ST.
To account for finite-sample uncertainty, we use the lower confidence bound
\begin{equation}
\gamma_{\mathrm{LCB}}
\;\triangleq\;
\max\!\left\{0,\; \mathrm{LCB}_{1-\alpha}(\hat a;n_{\mathrm{te}})-\tfrac12 \right\},
\end{equation}
where $\mathrm{LCB}_{1-\alpha}(\hat a;n_{\mathrm{te}})$ denotes the
$(1-\alpha)$ Wilson lower confidence bound on the accuracy computed from
$n_{\mathrm{te}}$ test samples.
The quantity $\gamma_{\mathrm{LCB}}$ provides a conservative, sample-size-aware
estimate of the classifier-detectable discrepancy between $P$ and $Q$.

\paragraph{Lower threshold: $\gamma_{\mathrm{small}}$ (minimum localizable effect size).}
The lower threshold $\gamma_{\mathrm{small}}$ defines the minimum global shift magnitude that reliably exceeds statistical noise and therefore warrants localization. We set
$\gamma_{\mathrm{small}}=0.01.$ Under the null hypothesis $P=Q$, the expected classifier accuracy is $0.5$, and
for typical test sizes used in our experiments ($N_{\mathrm{te}}\approx 20{,}000$),
the $95\%$ confidence interval half-width is approximately $0.005$. Consequently, an observed accuracy of $\hat a\approx0.51$ lies just beyond the noise floor and represents the smallest shift that can be distinguished from chance with high confidence.

\paragraph{Upper threshold: $\gamma_{\mathrm{large}}$ (saturation regime).}
When the classifier advantage is large, the distributions are broadly separable and most neighborhoods exhibit systematic imbalance. In this regime, local analysis provides limited additional information and may produce diffuse, difficult-to-interpret regions. Empirically, classifier accuracies around $\hat a\approx0.6$
(i.e., $\gamma=0.1$) correspond to cases where $P$ and $Q$ are easily separable,
indicating widespread distributional differences rather than localized anomalies.
In this regime, nearly all neighborhoods exhibit strong imbalance, and the
classifier signal saturates, making fine-grained localization unstable and less interpretable.

\paragraph{Band-pass rationale.}
Together, $\gamma_{\mathrm{small}}$ and $\gamma_{\mathrm{large}}$ define a
band-pass region
\[
\gamma_{\mathrm{small}} \;\le\; \gamma_{\mathrm{LCB}} \;\le\; \gamma_{\mathrm{large}},
\]
within which localized two-sample testing is both statistically justified and
practically informative.
By basing the decision on $\gamma_{\mathrm{LCB}}$, the proposed gating mechanism
adapts naturally to test-set size and controls false localization under negligible
shifts, while avoiding over-interpretation when the shift is overwhelmingly
global.

\subsection{Local Shift Detection}
\label{app:local_shift_detection}

When localization is triggered by the global test, we perform local
testing on a balanced held-out dataset
\[
\mathcal D_{\mathrm{loc}}
=
\{(Z_i,Y_i)\}_{i=1}^{N},
\]
where $Y_i=0$ denotes a source sample from $P_Z$ and $Y_i=1$
denotes a target sample from $Q_Z$. 

\paragraph{Neighborhood posterior and effect size.}
For each candidate point $Z_i$, let $\mathcal N_k(i)$ denote its
$k$ nearest neighbors in the latent space. We compute the smoothed
classifier posterior
\begin{equation*}
    \bar h_i
    =
    \frac{1}{k}
    \sum_{j\in\mathcal N_k(i)}
    h(Z_j).
\end{equation*}
To stabilize the log-odds calculation, define
\[
    \widetilde h_i
    =
    \operatorname{clip}
    \left(
        \bar h_i,\epsilon_h,1-\epsilon_h
    \right),
\]
where $\epsilon_h>0$ is a small constant. The local effect-size score is
\begin{equation}
    s_i
    =
    \left|
        \log
        \frac{\widetilde h_i}
             {1-\widetilde h_i}
    \right|.
    \label{eq:local_effect_size}
\end{equation}
The sign of $\bar h_i-1/2$ indicates the direction of the shift:
positive values correspond to target-heavy neighborhoods and negative
values to source-heavy neighborhoods.

\paragraph{Local significance test.}
Let $ c_i = \sum_{j\in\mathcal N_k(i)}Y_j$
denote the number of target samples in the neighborhood. Under the local-alignment null hypothesis
\[
    H_{0,i}:
    \mathbb P
    \left(
        Y=1\mid Z\in\mathcal N_k(i)
    \right)
    =
    \frac12,
\]
we use the working model
\[
    C_i\sim\operatorname{Binomial}
    \left(k,\frac12\right).
\]
The corresponding two-sided $p$-value is
\begin{equation}
    p_i
    =
    \mathbb P
    \left(
        \left|C_i-\frac{k}{2}\right|
        \geq
        \left|c_i-\frac{k}{2}\right|
    \right).
    \label{eq:local_binomial_pvalue}
\end{equation}

\paragraph{Neighborhood support.}
Reliable density-ratio estimation requires both distributions to have
sufficient support within a selected region. Let
$m_{\min}=\lceil\epsilon k\rceil$, where
$\epsilon\in[0,0.05]$, and define $\mathcal I
    =
    \left\{
        i:
        m_{\min}
        \leq c_i
        \leq k-m_{\min}
    \right\}.
$ Significant neighborhoods that fail this condition are treated as local
overlap violations and are excluded from importance-weight estimation.

\paragraph{Multiple-testing correction.}
We apply the Benjamini--Hochberg procedure to
$\{p_i\}_{i\in\mathcal I}$ at level $\alpha_{\mathrm{loc}}$.
Let $q_i$ denote the adjusted $p$-value and define
$ \mathcal I_{\alpha}
    =
    \left\{
        i\in\mathcal I:
        q_i\leq\alpha_{\mathrm{loc}}
    \right\}.$
The accepted candidates are ranked by the effect-size score
$s_i$ in Eq.~\eqref{eq:local_effect_size}.

\paragraph{Region selection and de-duplication.}
Let $r_i$ denote the distance from $Z_i$ to its $k$th nearest
neighbor. Starting from the highest-ranked candidate, we greedily
select a center and suppress any remaining candidate $j$ satisfying
\begin{equation}
    \|Z_j-Z_i\|_2
    \leq
    \lambda r_i,
    \label{eq:local_nms}
\end{equation}
where $\lambda>0$ controls the suppression radius. Let
$\mathcal C=\{i_1,\ldots,i_K\}$ denote the surviving centers and define
the preliminary regions
\[
    \widetilde R_k
    =
    \left\{
        Z_j:
        j\in\mathcal N_k(i_k)
    \right\}.
\]
Their union defines the detected mismatch set
\[
    \widehat S
    =
    \bigcup_{k=1}^{K}\widetilde R_k.
\]
When preliminary regions overlap, each sample in $\widehat S$ is
assigned to its nearest selected center:
\begin{equation*}
    R_k
    =
    \left\{
        z\in\widehat S:
        k
        =
        \arg\min_{\ell\in\{1,\ldots,K\}}
        \|z-Z_{i_\ell}\|_2
    \right\}.
\end{equation*}
The resulting regions are disjoint and satisfy
\[
    \widehat S
    =
    \bigcup_{k=1}^{K}R_k,
    \qquad
    R_k\cap R_\ell=\varnothing
    \quad
    (k\neq\ell),
\]
providing the regions used for localized importance weighting in
Section~\ref{sec:lis}.

\subsection{Proofs for Latent Metric Calibration}
\label{app:lis_theory}

This section proves the latent calibration identity and the asymptotic
MSE result for oracle normalized LIS. Throughout, let
\[
w^\star(z)
=
\frac{dQ_Z}{dP_Z}(z),
\qquad
\mu_Q
=
\mathbb E_{Q_X}[\ell(X)].
\]
For the second-order MSE expansion, we additionally assume
\[
0\leq \ell(X)\leq M,
0<\underline w
\leq
w^\star(Z)
\leq B<\infty,
P_Z\text{-almost surely}.
\]

\paragraph{Proof of the latent calibration identity.}

\begin{proof}[Proof of Theorem~\ref{thm:latent_calibration}]
Under Assumption~\ref{ass:latent_shift}, define
\[
m(z)
=
\mathbb E_{P_X}[\ell(X)\mid Z=z]
=
\mathbb E_{Q_X}[\ell(X)\mid Z=z].
\]
By the tower property and the Radon--Nikodym change-of-measure
identity,
\begin{align*}
\mu_Q
&=
\mathbb E_{Q_Z}[m(Z)]
\notag\\
&=
\mathbb E_{P_Z}
\left[
w^\star(Z)m(Z)
\right]
\notag\\
&=
\mathbb E_{P_X}
\left[
w^\star(Z)\ell(X)
\right].
\end{align*}

Suppose further that the shift is confined to
$\mathcal{S}\subseteq\mathcal Z$, so that
\[
P_Z(A)=Q_Z(A)
\qquad
\text{for every measurable }A\subseteq \mathcal{S}^c.
\]
The restrictions of $P_Z$ and $Q_Z$ to $S^c$ are therefore equal,
which implies
\[
w^\star(z)=1
\qquad
P_Z\text{-almost everywhere on }\mathcal{S}^c.
\]
Splitting the expectation over $\mathcal{S}$ and $\mathcal{S}^c$ yields
\begin{align}
\mu_Q
={}&
\mathbb E_{P_X}
\left[
\ell(X)\mathbb{I}\{Z\notin \mathcal{S}\}
\right]
\notag\\
&+
\mathbb E_{P_X}
\left[
w^\star(Z)\ell(X)
\mathbb{I}\{Z\in \mathcal{S}\}
\right],
\end{align}
which proves the localized calibration identity.
\end{proof}

\paragraph{Proof of the oracle normalized LIS result.}

\begin{proof}[Proof of Theorem~\ref{thm:oracle_lis_mse}]
Let
\[
X_1,\ldots,X_n
\stackrel{\mathrm{i.i.d.}}{\sim}P_X,
\qquad
Z_i=\phi(X_i),
\]
and define
\[
U_i
=
w^\star(Z_i)
\bigl(\ell(X_i)-\mu_Q\bigr),
\qquad
V_i
=
w^\star(Z_i)-1.
\]
By Theorem~\ref{thm:latent_calibration},
\[
\mathbb E_{P_X}[U_i]=0,
\qquad
\mathbb E_{P_Z}[V_i]=0.
\]
Let
\[
\overline U_n
=
\frac{1}{n}\sum_{i=1}^{n}U_i,
\qquad
\overline V_n
=
\frac{1}{n}\sum_{i=1}^{n}V_i.
\]
The oracle normalized LIS error can then be written exactly as
\begin{equation}
\widehat{\mu}_{Q,\mathrm{or}}^{\mathrm{LIS}}
-\mu_Q
=
\frac{\overline U_n}
     {1+\overline V_n}.
\label{eq:oracle_lis_ratio_error}
\end{equation}

Define
\[
\sigma_{\mathrm{LIS}}^2
=
\operatorname{Var}_{P_X}
\left[
w^\star(Z)
\bigl(\ell(X)-\mu_Q\bigr)
\right]
=
\mathbb E_{P_X}[U_i^2].
\]
Because $U_i$ has finite variance, the central limit theorem gives
\[
\sqrt n\,\overline U_n
\xrightarrow{d}
\mathcal N(0,\sigma_{\mathrm{LIS}}^2).
\]
Moreover, the law of large numbers implies
$\overline V_n\to 0$ in probability. Applying Slutsky's theorem to
Eq.~\eqref{eq:oracle_lis_ratio_error} gives
\[
\sqrt n
\left(
\widehat{\mu}_{Q,\mathrm{or}}^{\mathrm{LIS}}
-\mu_Q
\right)
\xrightarrow{d}
\mathcal N(0,\sigma_{\mathrm{LIS}}^2).
\]

We next establish the second-order MSE expansion. From
Eq.~\eqref{eq:oracle_lis_ratio_error},
\begin{equation}
\operatorname{MSE}_Q
\left(
\widehat{\mu}_{Q,\mathrm{or}}^{\mathrm{LIS}}
\right)
=
\mathbb E
\left[
\frac{\overline U_n^2}
     {(1+\overline V_n)^2}
\right].
\label{eq:oracle_lis_mse_ratio}
\end{equation}
The lower weight bound implies
\[
1+\overline V_n
=
\frac{1}{n}
\sum_{i=1}^{n}w^\star(Z_i)
\geq
\underline w,
\]
so the random denominator is uniformly bounded away from zero.

For
\[
f(v)=(1+v)^{-2},
\]
a second-order Taylor expansion around $v=0$ gives
\[
f(v)
=
1-2v+3v^2+R(v),
\qquad
|R(v)|\leq C_f|v|^3,
\]
where $C_f<\infty$ depends only on
$\underline w$ and $B$. Therefore,
\begin{align}
\operatorname{MSE}_Q
\left(
\widehat{\mu}_{Q,\mathrm{or}}^{\mathrm{LIS}}
\right)
={}&
\mathbb E[\overline U_n^2]
-
2\mathbb E[\overline U_n^2\overline V_n]
\notag\\
&+
3\mathbb E[\overline U_n^2\overline V_n^2]
+
\mathbb E[\overline U_n^2R(\overline V_n)].
\label{eq:oracle_lis_taylor}
\end{align}
Because $U_i$ and $V_i$ are centered, bounded, and i.i.d.,
\[
\mathbb E[\overline U_n^2]
=
\frac{\sigma_{\mathrm{LIS}}^2}{n},
\]
\[
\mathbb E[\overline U_n^2\overline V_n]
=
\frac{\mathbb E[U_i^2V_i]}{n^2},
\]
and standard moment counting yields
\[
\mathbb E[\overline U_n^2\overline V_n^2]
=
O\left(\frac{1}{n^2}\right),
\qquad
\mathbb E
\left[
\overline U_n^2
|R(\overline V_n)|
\right]
=
O\left(\frac{1}{n^{5/2}}\right).
\]
Substituting these relations into
Eq.~\eqref{eq:oracle_lis_taylor} gives
\begin{equation}
\operatorname{MSE}_Q
\left(
\widehat{\mu}_{Q,\mathrm{or}}^{\mathrm{LIS}}
\right)
=
\frac{\sigma_{\mathrm{LIS}}^2}{n}
+
O\left(\frac{1}{n^2}\right).
\label{eq:appendix_oracle_lis_mse}
\end{equation}

For the naive estimator,
\[
\widehat{\mu}_{P}^{\mathrm{naive}}
=
\frac{1}{n}
\sum_{i=1}^{n}\ell(X_i),
\]
let
\[
\mu_P
=
\mathbb E_{P_X}[\ell(X)],
\Delta
=
|\mu_P-\mu_Q|,
\sigma_P^2
=
\operatorname{Var}_{P_X}[\ell(X)].
\]
Its MSE relative to $\mu_Q$ is exactly
\begin{equation}
\operatorname{MSE}_Q
\left(
\widehat{\mu}_{P}^{\mathrm{naive}}
\right)
=
\Delta^2
+
\frac{\sigma_P^2}{n}.
\label{eq:appendix_naive_mse}
\end{equation}
Combining Eqs.~\eqref{eq:appendix_oracle_lis_mse}
and~\eqref{eq:appendix_naive_mse} gives
\begin{align}
&
\operatorname{MSE}_Q
\left(
\widehat{\mu}_{P}^{\mathrm{naive}}
\right)
-
\operatorname{MSE}_Q
\left(
\widehat{\mu}_{Q,\mathrm{or}}^{\mathrm{LIS}}
\right)
\notag\\
&\qquad
=
\Delta^2
-
\frac{
\sigma_{\mathrm{LIS}}^2-\sigma_P^2
}{n}
+
O\left(\frac{1}{n^2}\right).
\label{eq:appendix_lis_mse_difference}
\end{align}

More formally, there exist constants $C_{\mathrm{LIS}}<\infty$
and $n_0$ such that, for all $n\geq n_0$,
\[
\left|
\operatorname{MSE}_Q
\left(
\widehat{\mu}_{Q,\mathrm{or}}^{\mathrm{LIS}}
\right)
-
\frac{\sigma_{\mathrm{LIS}}^2}{n}
\right|
\leq
\frac{C_{\mathrm{LIS}}}{n^2}.
\]
Hence, oracle normalized LIS has a smaller MSE whenever
\begin{equation}
\Delta^2
>
\frac{
\sigma_{\mathrm{LIS}}^2-\sigma_P^2
}{n}
+
\frac{C_{\mathrm{LIS}}}{n^2}.
\label{eq:appendix_lis_dominance}
\end{equation}
\end{proof}

\paragraph{Connection to the practical LIS estimator.}

The preceding result assumes oracle knowledge of the true density ratio
$w^\star$ and the true mismatch region $S$. In practice, they are
replaced by the detected region $\widehat S$ and the estimated,
piecewise-constant weight $\widehat w_{\mathrm{loc}}$.

On the same evaluation sample, write
\[
w_i^\star=w^\star(Z_i),
\qquad
\widehat w_i
=
\widehat w_{\mathrm{loc}}(Z_i).
\]
Then the practical estimation error admits the exact decomposition
\begin{align}
\widehat{\mu}_{Q}^{\mathrm{LIS}}-\mu_Q
={}&
\left(
\widehat{\mu}_{Q}^{\mathrm{LIS}}
-
\widehat{\mu}_{Q,\mathrm{or}}^{\mathrm{LIS}}
\right)
\notag\\
&+
\left(
\widehat{\mu}_{Q,\mathrm{or}}^{\mathrm{LIS}}
-\mu_Q
\right).
\label{eq:practical_oracle_decomposition}
\end{align}
The first term captures region-detection, density-ratio estimation,
and clipping errors, while the second is the oracle sampling error
analyzed above.

Provided that $\sum_i\widehat w_i>0$, the difference between the
practical and oracle estimators satisfies the deterministic bound
\begin{equation}
\left|
\widehat{\mu}_{Q}^{\mathrm{LIS}}
-
\widehat{\mu}_{Q,\mathrm{or}}^{\mathrm{LIS}}
\right|
\leq
2M
\frac{
\sum_{i=1}^{n}
|\widehat w_i-w_i^\star|
}{
\sum_{i=1}^{n}\widehat w_i
}.
\label{eq:practical_oracle_weight_bound}
\end{equation}
Thus, the practical estimator approaches its oracle counterpart when
the estimated localized weights converge to the true density ratio and
their empirical normalization remains bounded away from zero. The
oracle theorem isolates the sampling and self-normalization uncertainty;
a complete finite-sample analysis of the practical estimator would
additionally require rates for region detection and density-ratio
estimation.

\section{Additional Experiments}
 
\subsection{1D Synthetic Experiment: Sensitivity of Shift Detection}
To evaluate the classifier-detectable discrepancy
$2\widehat\gamma$ and compare its sensitivity with standard
two-sample statistics, we construct a controlled one-dimensional
benchmark. We specifically investigate how different types of distributional shifts impact the evaluation of a binary performance metric.

\paragraph{Data Generation and Performance Metric}
The reference distribution $P$ and three target distributions $Q$ are generated using Gaussian Mixture Models (GMMs). For each case, we sample $n = 20,000$ points. The performance metric is defined as the indicator function $\ell(x) = \mathbb{I}(x > 1.5)$. The scenarios are constructed as follows:
\begin{itemize}
    \item \textbf{Case 1 (No Shift):} $Q$ is identical to $P$, where $P \sim 0.7\mathcal{N}(0, 1) + 0.3\mathcal{N}(4, 0.8)$.
    \item \textbf{Case 2 (Localized Difference):} $Q$ introduces a localized ``bump'' of contamination via an $\epsilon$-mixture ($10\%$):
    \[ Q \sim (1 - \epsilon)P + \epsilon\mathcal{N}(2.0, 0.12). \]
    This simulates a scenario where a specific environmental condition is over-represented.
    \item \textbf{Case 3 (Global Shift):} $Q$ represents a significant structural change in the context distribution:
    \[ Q \sim 0.4\mathcal{N}(-1.0, 1.2) + 0.6\mathcal{N}(6.0, 1.0). \]
\end{itemize}

\paragraph{Statistical Implementation}
We evaluate the shift using four distinct two-sample testing frameworks, using permutation tests ($N_{perm} \in [60, 80]$) to calculate $p$-values:

\begin{enumerate}
    \item \textbf{Classifier Two-Sample Test (C2ST):} We train a kNN classifier with a \texttt{StandardScaler} pipeline on a 50/50 train-test split. The test accuracy $\widehat{\text{acc}}$ provides the empirical basis for the shift bound $2\hat{\gamma}$, with the $p$-value determined via label permutation.
    \item \textbf{Kolmogorov-Smirnov (KS) Test:} We compute the two-sided KS statistic, defined as the maximum absolute difference between the empirical cumulative distribution functions (ECDFs) of $P$ and $Q$:
    \[ D_{n,m} = \sup_{x} |F_{n,P}(x) - F_{m,Q}(x)|. \]
    The $p$-value is calculated using the asymptotic distribution to assess the probability of the observed discrepancy under $H_0$.
    \item \textbf{MMD with Random Fourier Features (MMD-RFF):} We approximate the Gaussian kernel using $D=128$ random features. The kernel bandwidth $\sigma$ is determined via the \textit{median heuristic} on a sub-sample to adapt to the data scale.
    \item \textbf{Wasserstein-1 Distance ($W_1$):} We compute the 1D Earth Mover's Distance using the closed-form identity for sorted samples:
    \[ \widehat{W}_1(P, Q) = \frac{1}{n} \sum_{i=1}^n \big| x_{(i)} - y_{(i)} \big|. \]
\end{enumerate}



\subsection{Real-world dataset experiments}\label{app:real_exp}
\subsubsection{Cut-In Road-Testing Experiments}
\label{sec:real_world_exp}
\paragraph{Dataset.}
The cut-in dataset contains approximately 177,000 road-test segments, each spanning 10--15 seconds. These segments are extracted from long-horizon onboard logs and cover diverse cut-in geometries, traffic densities, vehicle types, and interaction patterns.

The representation-learning set consists of one month of data collected under multiple AV system versions. AV-version identifiers are not provided to the encoder. For downstream distribution comparison, all source--target pairs are constructed under a fixed AV system version so that the detected differences primarily reflect changes in scenario exposure rather than policy changes.
\paragraph{Representation learning.}
We train the Transformer-based encoder--decoder described in
Section~\ref{sec:ae_impl} to obtain a
64-dimensional scene representation. The model operates on
trajectory-level inputs extracted from real-world road-testing logs,
including synchronized ego-vehicle states, surrounding-agent states,
and road-context features. Each scenario contains up to 128 agents over
a temporal horizon of 15 steps.

The encoder comprises three main components. First, separate agent and
ego encoders project the raw inputs into a shared feature space with
$d_{\mathrm{model}}=128$. Second, two cross-attention modules capture
complementary ego--agent and ego--road interactions. Third, a temporal
Transformer encoder with five layers and eight attention heads
aggregates the interaction-aware features over time. The resulting
sequence is temporally pooled and projected into a 64-dimensional
scene-level latent embedding.

The decoder follows a Transformer-based architecture. The latent
embedding is first projected into temporal memory tokens and then
decoded to reconstruct behavior-relevant scenario information,
including all-agent trajectories, cut-in-agent features, ego--agent
relative features, and agent-validity masks. The reconstructed
trajectory output has dimension two per agent, with a hidden dimension
of 128. Additional MLP heads predict auxiliary dynamic attributes,
including time to collision (TTC) and other interaction-related
variables.

The model is optimized using
\[
\mathcal L_{\mathrm{train}}
=
\mathcal L_{\mathrm{recon}}
+
\lambda_{\mathrm{attr}}\mathcal L_{\mathrm{attr}},
\]
where $\mathcal L_{\mathrm{recon}}$ measures reconstruction error and
$\mathcal L_{\mathrm{attr}}$ supervises the prediction of
behavior-relevant attributes from the latent representation.

The model is trained using Adam with a learning rate of
$10^{-4}$, a batch size of 256, and 5,000 epochs. Training is performed
on a single NVIDIA H100 GPU and requires approximately two weeks. After
training, the learned latent representations are used for global
distribution comparison, localized shift detection, metric calibration,
and analysis across AV software versions.

\subsubsection{Visualization of the Learned Latent Space}
\label{app:vis_latent}

\begin{figure*}[t]
  \centering
  \includegraphics[width=\textwidth]
  {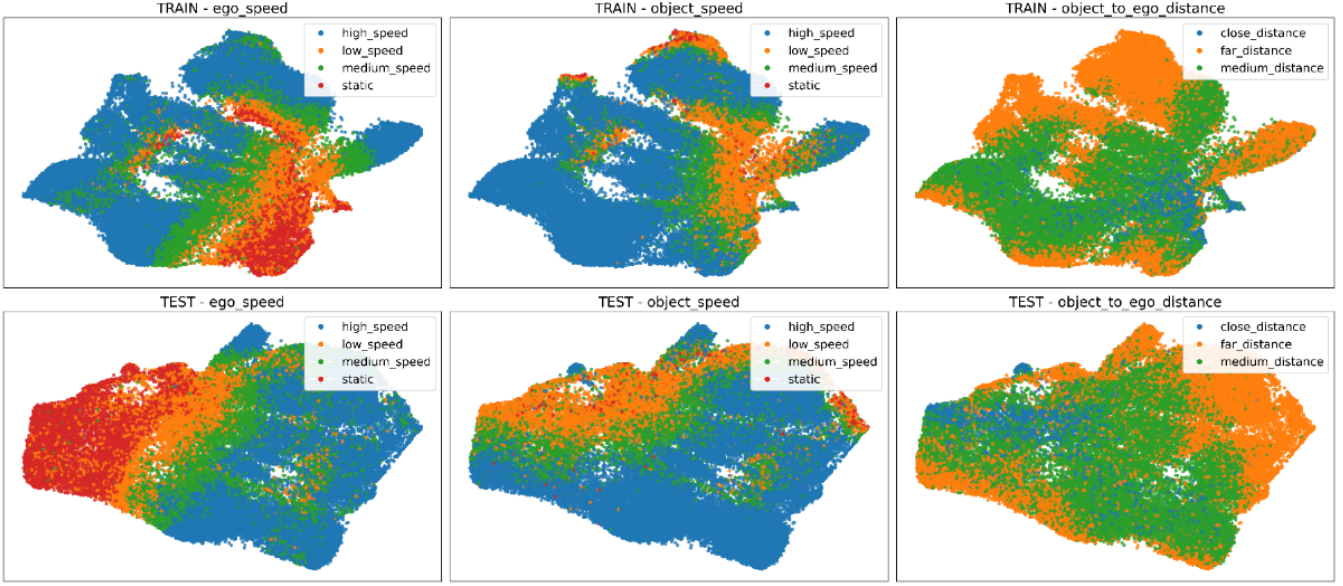}
  \caption{
  Two-dimensional visualization of the learned cut-in representation.
  The latent codes are reduced using PCA followed by UMAP. The top and
  bottom rows show training and held-out data, respectively, colored by
  ego speed, cut-in-agent speed, and relative distance.
  }
  \label{fig:latent_space_umap}
\end{figure*}

We qualitatively examine whether the learned representation preserves
behaviorally relevant scenario information. Each cut-in event is first
encoded into the 64-dimensional latent space, reduced to 30 dimensions
using PCA, and subsequently projected into two dimensions using UMAP.
A common projection pipeline is used for the training and held-out
datasets.

As shown in Fig.~\ref{fig:latent_space_umap}, scenarios with similar
ego speeds, cut-in-agent speeds, and relative distances occupy
structured regions of the latent space. These semantic attributes are
obtained through rule-based post-processing and manual verification;
they are used only for visualization and are not provided to the
encoder during training. The consistent organization observed across
the training and held-out splits provides qualitative evidence that the
encoder preserves behaviorally relevant contextual information.

For an additional qualitative evaluation, we randomly select anchor
scenarios and retrieve their nearest neighbors in the latent space.
The retrieved examples exhibit similar traffic configurations and
interaction patterns, including low-speed close-range cut-ins and
high-speed long-range cut-ins with one or more additional vehicles positioned between the ego vehicle and the cut-in agent, as shown in Fig.\ref{fig:latent_retrieval}.

\begin{figure*}[!h]
  \centering
  \includegraphics[width=\textwidth]{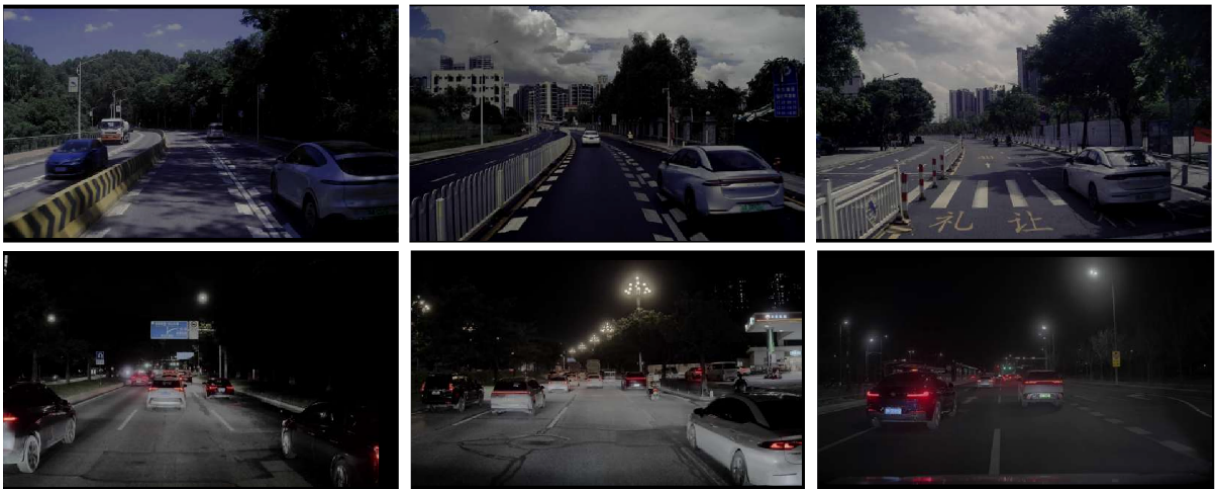}
  \caption{
  Nearest-neighbor retrieval examples. Each row shows a randomly
  selected anchor scenario followed by its nearest neighbors in the
  learned latent space.
  }
  \label{fig:latent_retrieval}
\end{figure*}

\paragraph{Cluster-level validation.}
We further cluster the cut-in embeddings into 27 groups and examine the distribution of events of interest across the resulting clusters. The event shares are clearly non-uniform across clusters, with a standard deviation of $3.04$ percentage points. Moreover, the cluster-level event rates exhibit substantial heterogeneity, with a standard deviation of $8.89$ percentage points. These results indicate that the learned representation organizes scenarios into interaction patterns associated with distinct empirical risk levels, rather than producing arbitrary or uniformly distributed partitions.
\begin{figure*}[!h]
  \centering
  \includegraphics[width=\textwidth]
  {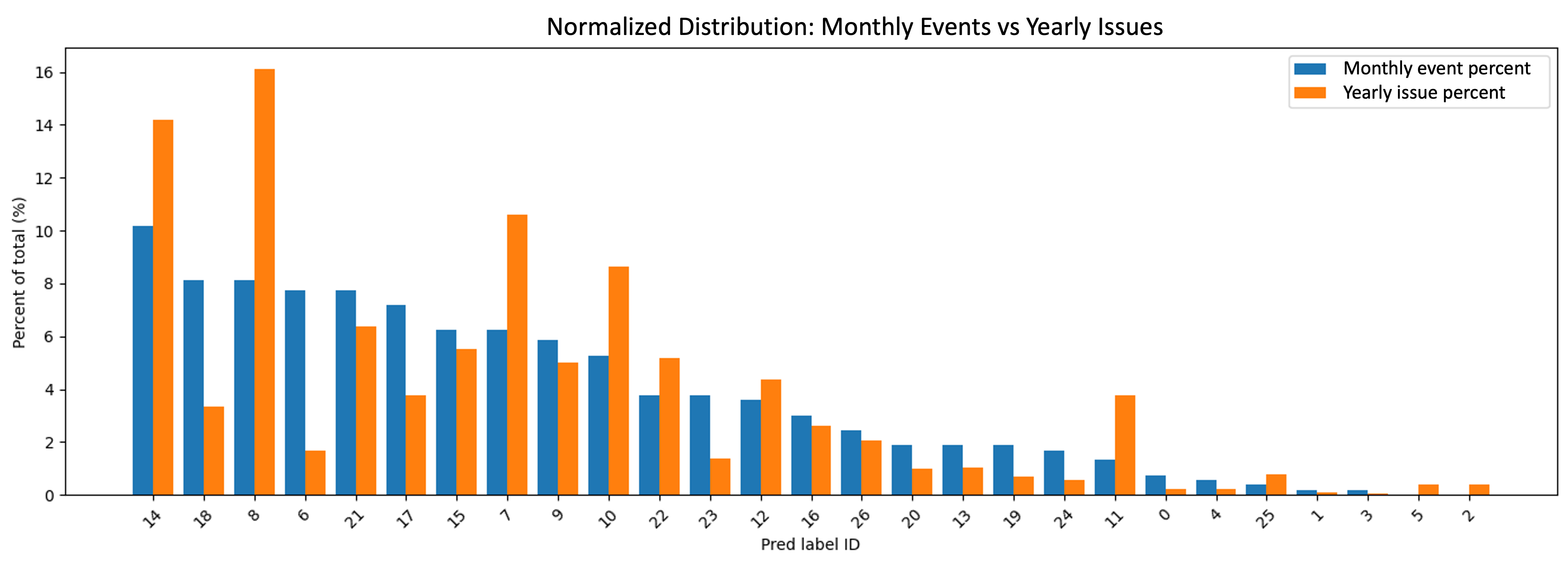}
  \caption{
  Normalized cluster-level distributions of monthly events and
  independently collected yearly issues across the 27 latent clusters.
  }
  \label{fig:cluster_event_correlation}
\end{figure*}
Based on manual semantic inspection, the clusters can be organized into four broad groups. The first group (clusters 0--5) comprises ego-turning scenarios, including right turns at T-junctions, waiting for or following a U-turning vehicle, and left- or right-turn maneuvers at different speed levels. The second group (clusters 6--15) captures close-range and high-speed interactions, including crossing or wrong-way agents, congested cut-ins, cut-ins from different relative directions, high-speed merging, and agents moving away from the ego vehicle. The third group (clusters 16--22) consists of congestion, intersection, and response related scenarios, such as obstacle-induced deceleration, queue entry, intersection crossing, surrounding-agent overtaking, and ego responses to nearby disturbances. The fourth group (clusters 23--25) contains merging and start-up behaviors, including vehicles entering from the left or right adjacent lane and previously stationary agents beginning to move. The final cluster (cluster 26) is retained as a general category for heterogeneous or weakly characterized scenarios that do not exhibit a single dominant interaction pattern. Overall, the discovered clusters capture complementary variations in road topology, relative position, speed, interaction distance, and
temporal evolution, providing an interpretable semantic structure for
the learned latent space.

The monthly event distribution (Fig.~\ref{fig:cluster_event_correlation}) is also strongly correlated with an
independently collected yearly issue distribution, achieving a Pearson
correlation of $r=0.758$ ($p=4.57\times10^{-6}$) and a Spearman rank
correlation of $\rho=0.840$ ($p=4.11\times10^{-8}$). Unlike the monthly events, which are automatically detected using predefined metrics, the yearly issues are reported and reviewed by human operators. This agreement
provides external evidence that the discovered clusters capture
recurring and temporally consistent interaction patterns.

\begin{figure}[!h]
  \centering

  \begin{subfigure}[t]{0.48\linewidth}
    \centering
    \includegraphics[width=\linewidth]
    {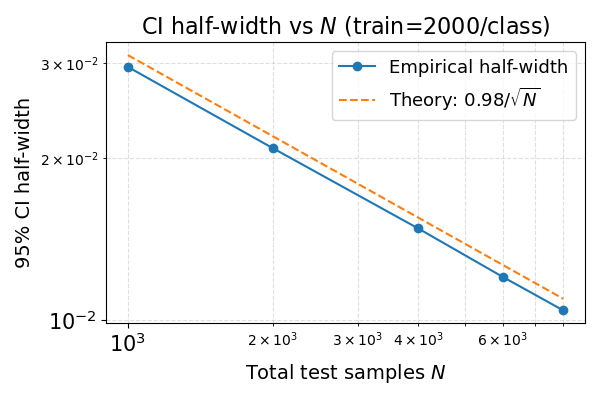}
    \caption{Cross-region: CI half-width versus test size.}
    \label{fig:c2st_sensitivity:a}
  \end{subfigure}
  \hfill
  \begin{subfigure}[t]{0.48\linewidth}
    \centering
    \includegraphics[width=\linewidth]
    {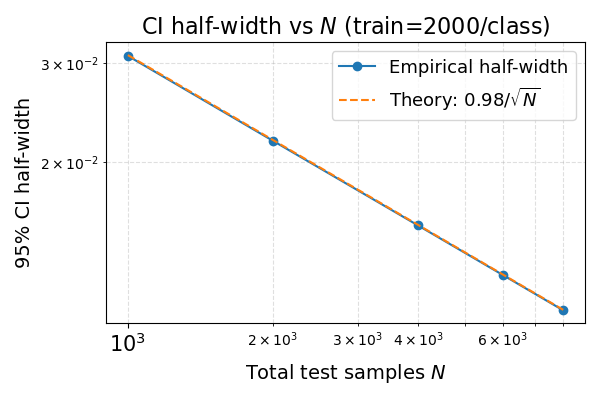}
    \caption{Cross-week: CI half-width versus test size.}
    \label{fig:c2st_sensitivity:b}
  \end{subfigure}

  \vspace{0.4em}

  \begin{subfigure}[t]{0.48\linewidth}
    \centering
    \includegraphics[width=\linewidth]
    {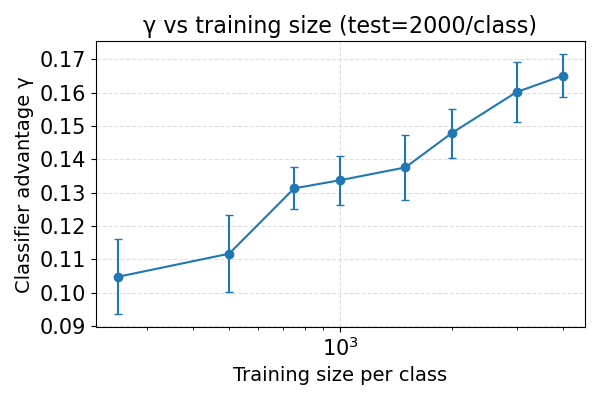}
    \caption{Cross-region: $\widehat{\gamma}$ versus training size.}
    \label{fig:c2st_sensitivity:c}
  \end{subfigure}
  \hfill
  \begin{subfigure}[t]{0.48\linewidth}
    \centering
    \includegraphics[width=\linewidth]
    {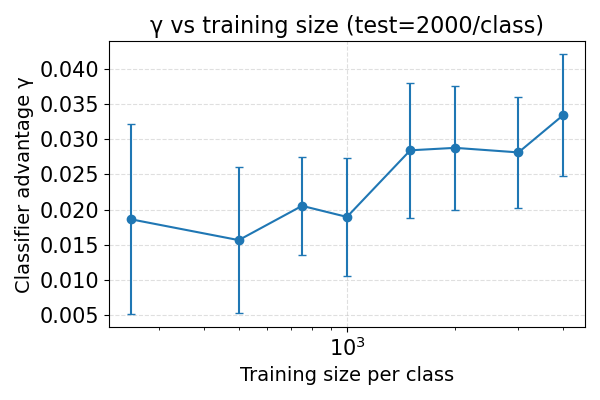}
    \caption{Cross-week: $\widehat{\gamma}$ versus training size.}
    \label{fig:c2st_sensitivity:d}
  \end{subfigure}

  \caption{
  Sensitivity of C2ST to the held-out test-set size and
  classifier-training size. The upper panels report empirical
  confidence-interval half-widths, while the lower panels report the
  mean classifier advantage with error bars across repeated runs.
  }
  \label{fig:c2st_sensitivity}
\end{figure}

\subsubsection{C2ST Implementation and Uncertainty}
\label{appendix:classifier_details}

The global C2ST classifier is a three-layer multilayer perceptron with
hidden dimensions 128 and 64, ReLU activations, and a scalar logit
output. It is trained using binary cross-entropy loss. The architecture
and optimization procedure are held fixed unless explicitly varied in
the robustness experiments.

The estimated classifier advantage is affected by both finite test-set
uncertainty and classifier-training variability. As shown in the upper
panels of Fig.~\ref{fig:c2st_sensitivity}, the empirical 95\%
confidence-interval half-width decreases at the expected
$\mathcal O(N_{\mathrm{te}}^{-1/2})$ rate for both cross-region and
cross-week comparisons.

To examine sensitivity to the classifier-training size, we vary the
number of training samples per class while fixing a balanced held-out
test set containing 2,000 samples per class, so that
$N_{\mathrm{te}}=4,000$. At this test-set size, the worst-case standard
deviation of the empirical accuracy is
\[
\sqrt{\frac{1}{4N_{\mathrm{te}}}}
\approx 0.0079,
\]
corresponding to a 95\% normal-approximation half-width of approximately
$0.0155$. The lower panels of Fig.~\ref{fig:c2st_sensitivity} therefore
reflect both classifier-training variability and the finite test-set
uncertainty quantified above.

\begin{figure*}[t]
    \centering
    \includegraphics[width=\textwidth]
    {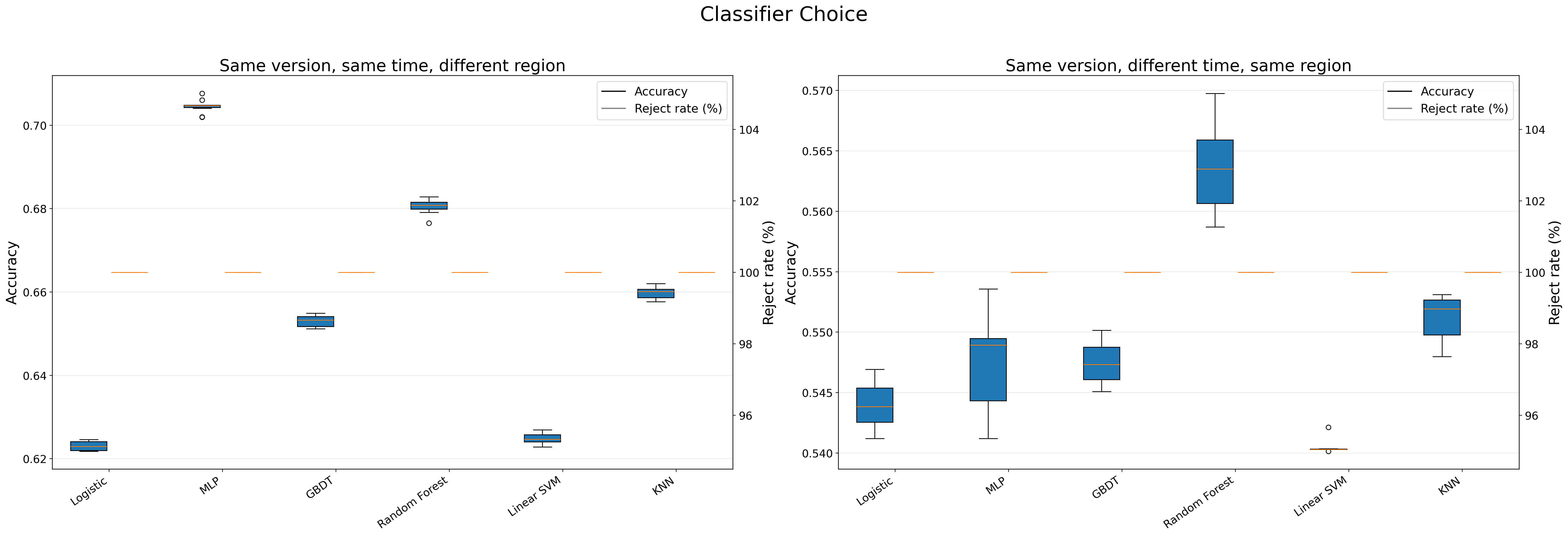}
    \caption{
    Sensitivity of global C2ST to the classifier family. Cross-region
    and cross-week comparisons are evaluated using logistic regression,
    MLP, GBDT, random forest, linear SVM, and kNN classifiers. The
    boxplots summarize variability across repeated runs.
    }
    \label{fig:robust_classifier_choice}
\end{figure*}

\begin{figure*}[!h]
    \centering
    \includegraphics[width=\textwidth]
    {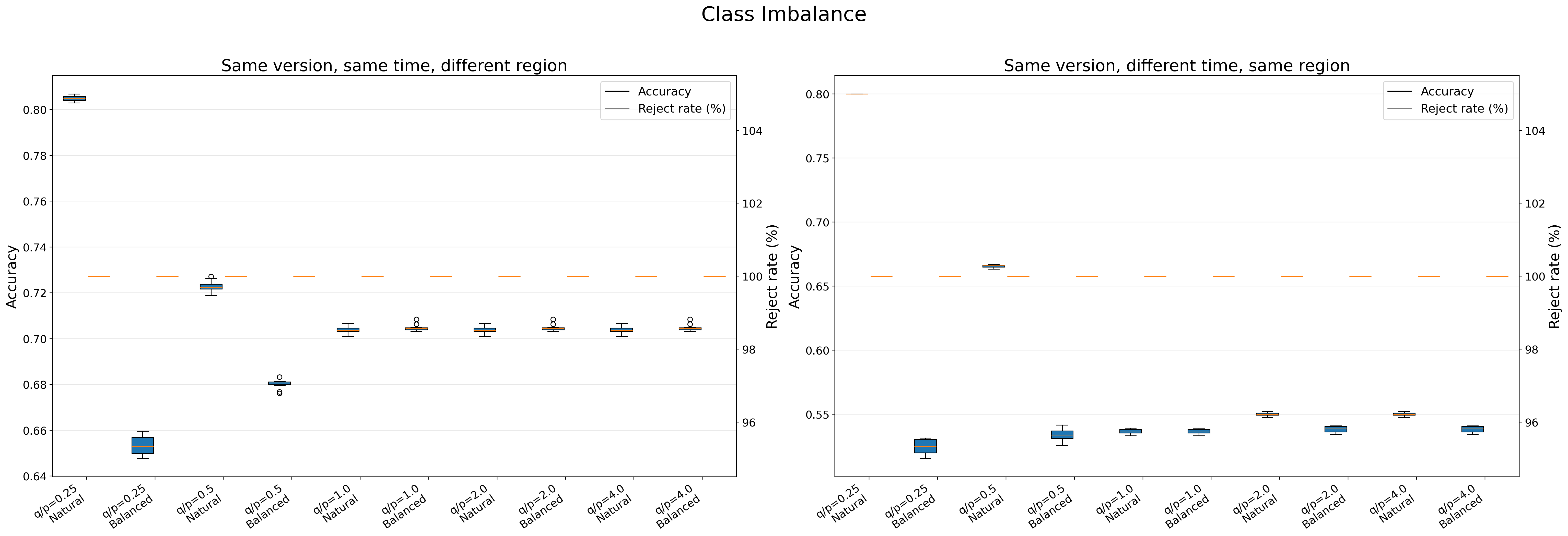}
    \caption{
    Sensitivity of global C2ST to the source--target class ratio in the
    classifier-training set. Natural sampling and balanced subsampling
    are compared, while accuracy is evaluated on the same balanced
    held-out test set.
    }
    \label{fig:robust_class_imbalance}
\end{figure*}

\begin{figure*}[!h]
    \centering
    \includegraphics[width=\textwidth]
    {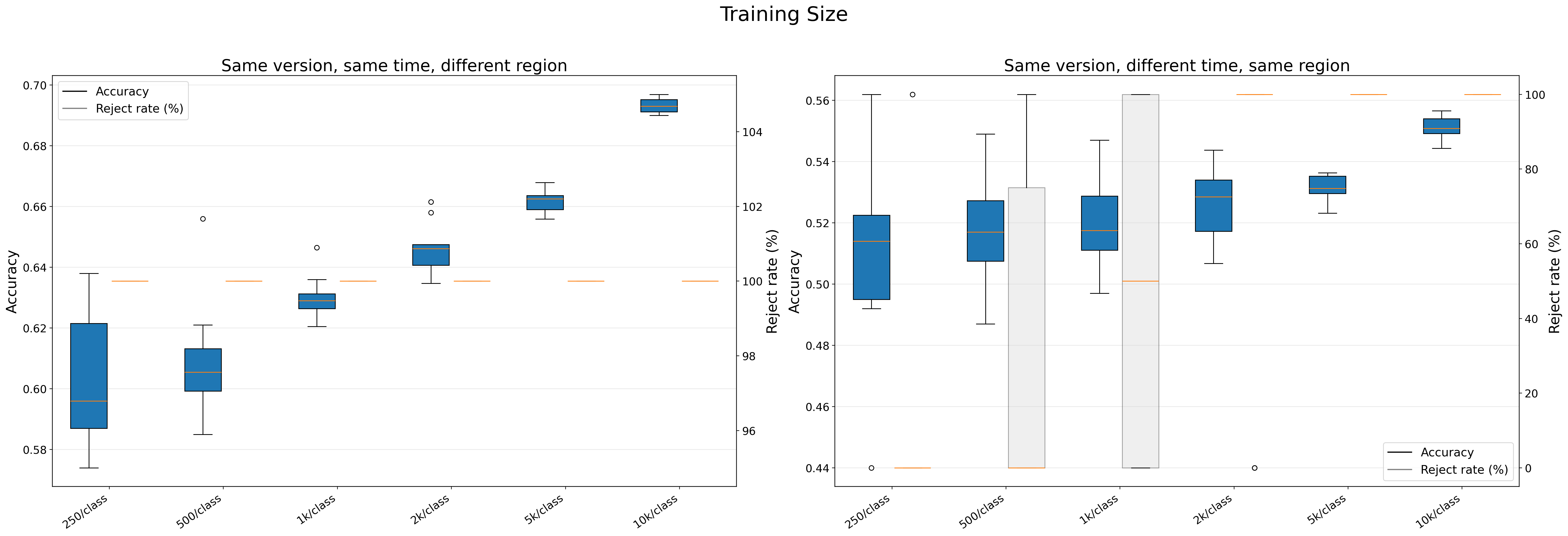}
    \caption{
    Sensitivity of global C2ST to the number of classifier-training
    samples per class. The stronger cross-region shift remains
    comparatively stable, whereas the weaker cross-week shift is more
    sensitive to limited training data.
    }
    \label{fig:robust_training_size}
\end{figure*}

\begin{figure*}[!h]
    \centering
    \includegraphics[width=\textwidth]
    {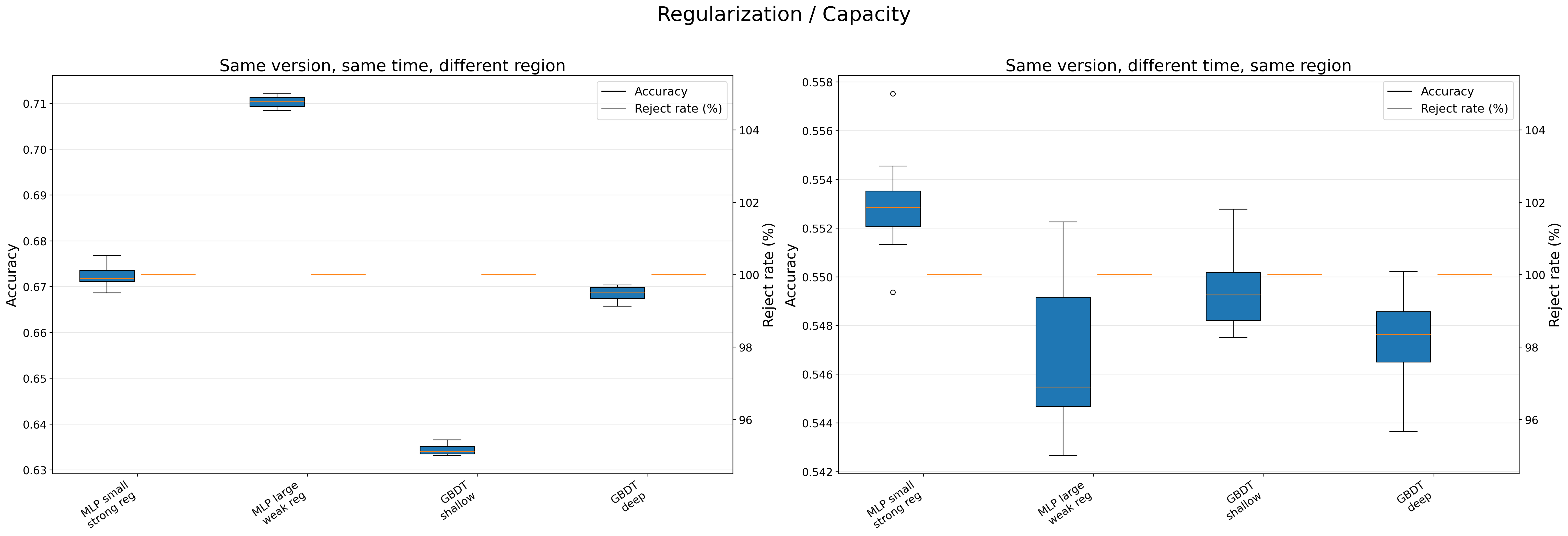}
    \caption{
    Sensitivity of global C2ST to classifier capacity and
    regularization. The evaluated configurations include small and
    large MLPs, together with shallow and deep GBDT classifiers.
    }
    \label{fig:robust_regularization_capacity}
\end{figure*}

\paragraph{Robustness results.}
We evaluate the robustness of global C2ST with respect to classifier
family, source--target class balance, classifier-training size, and
classifier capacity and regularization.
Table~\ref{tab:robustness_summary} summarizes the results across all
tested configurations, while
Figs.~\ref{fig:robust_classifier_choice}--\ref{fig:robust_regularization_capacity}
provide detailed results for the individual sensitivity factors.

\begin{table}[!h]
\centering
\caption{
Robustness of global C2ST under cross-region and cross-week
comparisons. $\#$ Configs denotes the number of experimental
configurations evaluated for each sensitivity factor. ACC is reported
as the mean $\pm$ standard deviation across the corresponding
configurations and random seeds.
}
\label{tab:robustness_summary}
\small
\setlength{\tabcolsep}{4pt}
\begin{tabular}{llcc}
\toprule
Regime
& Sensitivity factor
& $\#$ Configs
& ACC \\
\midrule

\multirow{4}{*}{Cross-region}
& Classifier family
& 6
& $0.658 \pm 0.029$ \\

& Training class ratio
& 10
& $0.709 \pm 0.037$ \\

& Training-set size
& 6
& $0.641 \pm 0.033$ \\

& Classifier capacity/reg.
& 4
& $0.671 \pm 0.027$ \\

\midrule

\multirow{4}{*}{Cross-week}
& Classifier family
& 6
& $0.549 \pm 0.008$ \\

& Training class ratio
& 10
& $0.577 \pm 0.084$ \\

& Training-set size
& 6
& $0.526 \pm 0.021$ \\

& Classifier capacity/reg.
& 4
& $0.549 \pm 0.003$ \\

\bottomrule
\end{tabular}
\end{table}

As summarized in Table~\ref{tab:robustness_summary}, the qualitative
ordering of the two comparison regimes remains stable across all tested
sensitivity factors: cross-region comparisons consistently yield higher
classification accuracy than cross-week comparisons. This observation
supports the main-text conclusion that geographic variation induces a
stronger distributional shift than temporal variation within the same
region.

Figure~\ref{fig:robust_classifier_choice} shows that the numerical C2ST
accuracy varies across classifier families, reflecting differences in
the portions of the latent discrepancy captured by their decision
functions. Nevertheless, all classifier families preserve the same
qualitative ordering between the cross-region and cross-week regimes.

Training-class imbalance introduces greater variability, particularly
for the weaker cross-week comparison, as shown in
Fig.~\ref{fig:robust_class_imbalance}. However, when performance is
evaluated on a balanced held-out test set, the cross-region discrepancy
remains consistently stronger. Among the evaluated factors,
classifier-training size has the most pronounced effect
(Fig.~\ref{fig:robust_training_size}). Limited training data reduce the
estimated accuracy substantially for the cross-week comparison, whose
signal lies closer to the detection threshold, whereas the stronger
cross-region shift remains comparatively stable. Finally,
Fig.~\ref{fig:robust_regularization_capacity} shows that classifier
capacity and regularization affect the magnitude of the estimated
discrepancy but do not alter the qualitative ordering between the two
regimes.
\begin{table*}[!t]
\centering
\caption{
Full-scene distribution comparison and metric-calibration results under
a fixed AV system version. Day~1 is treated as the source distribution,
while Days~2 and~3 serve as the target distributions. ACC and
$2\widehat{\gamma}=2(\mathrm{ACC}-0.5)$ quantify the global
classifier-detectable discrepancy. Relative error is computed between
the LIS-calibrated estimate and the corresponding target reference
value. }
\label{tab:full_scene_calibration}
\small
\setlength{\tabcolsep}{5pt}
\begin{tabular}{lcccccc}
\toprule
Comparison
& ACC
& $2\widehat{\gamma}$
& Naive estimate
& Target reference
& LIS estimate
& Relative error (\%) \\
\midrule
Day~1 $\rightarrow$ Day~2
& $0.536$
& $0.072$
& $0.011001$
& $0.011485$
& $0.011412$
& $0.63$ \\

Day~1 $\rightarrow$ Day~3
& $0.542$
& $0.084$
& $0.011001$
& $0.011764$
& $0.011650$
& $0.97$ \\
\bottomrule
\end{tabular}
\end{table*}
\subsubsection{Localization of Real-World Distribution Shifts}
\label{app:real_localization}
Figure~\ref{fig:global_real_comparison} provides a qualitative
visualization of the two real-world comparison settings analyzed in the
main text. The cross-region distributions exhibit more pronounced
changes in latent-space occupancy, whereas the cross-week distributions
retain substantially greater overlap. This visual pattern is consistent
with both the main-text C2ST results and the robustness analysis in
Table~\ref{tab:robustness_summary}. All statistical tests and local
difference analyses are performed in the original latent space; the
two-dimensional projections are used only for qualitative
interpretation.
\begin{figure}[t]
    \centering
    \includegraphics[width=0.99\linewidth]
    {figures/same_version_diff_region_three_panels_paper.pdf}

    \vspace{0.3em}

    \includegraphics[width=0.99\linewidth]
    {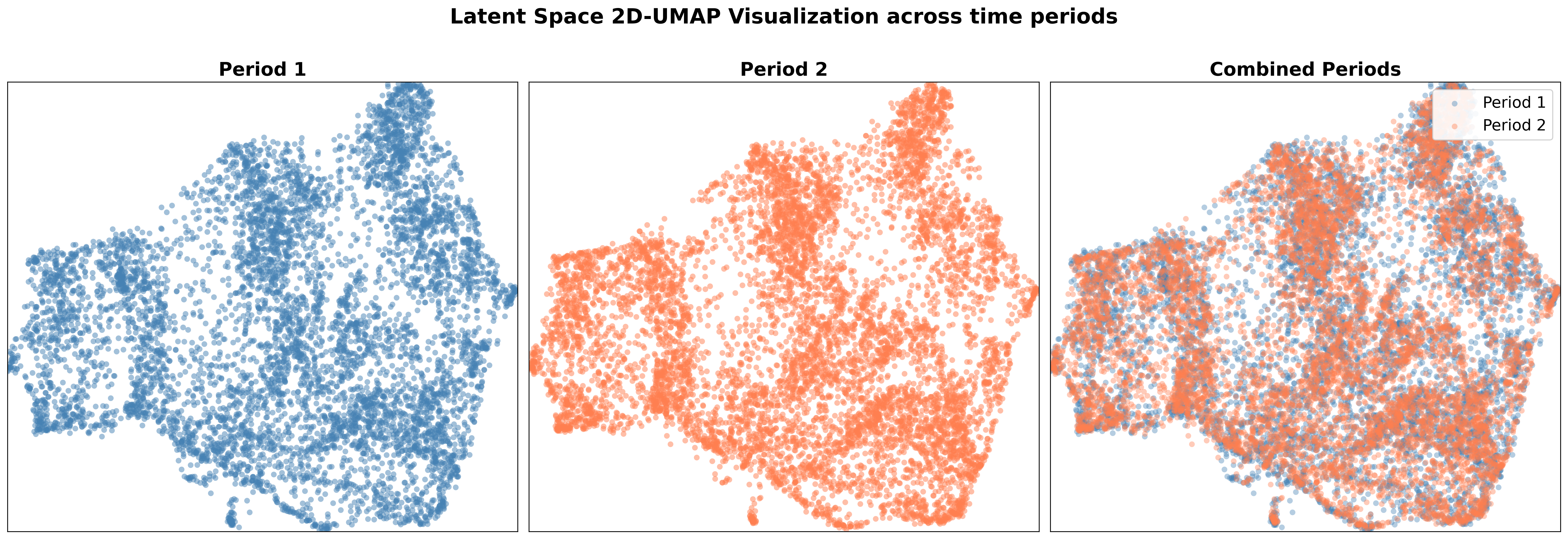}

    \caption{
    Qualitative visualization of the real-world latent distributions.
    Top: cross-region comparison under the same AV system version and
    collection period. Bottom: cross-week comparison within the same
    region and under the same AV system version. The cross-region
    distributions exhibit more pronounced changes in latent-space
    occupancy, whereas the cross-week distributions show greater
    overlap. The projections are used only for visualization; all
    statistical analyses are conducted in the original latent space.
    }
    \label{fig:global_real_comparison}
\end{figure}
For the moderate cross-week shift, the second stage of the proposed
framework identifies the latent neighborhoods that contribute most
strongly to the global discrepancy. Representative scenarios sampled
from the detected mismatch regions are shown in
Fig.~\ref{fig:local_region_examples}.
\begin{figure}[!h]
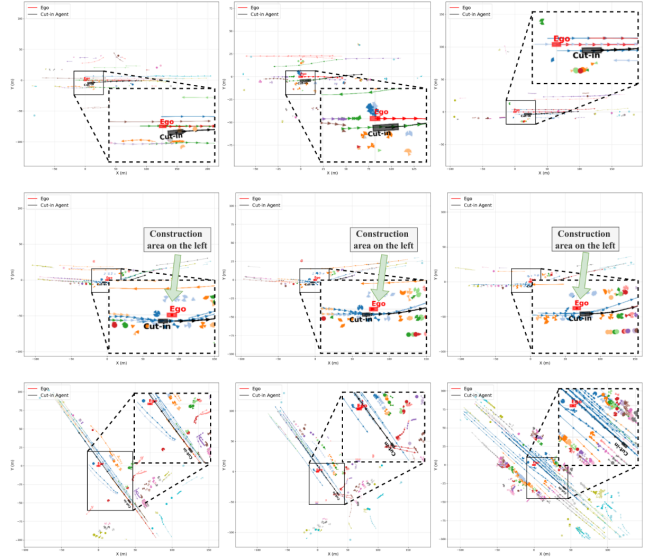

\centering
\includegraphics[width=\linewidth]{figures/cluster1.pdf}

\vspace{0.4em}
\includegraphics[width=\linewidth]{figures/cluster2.pdf}

\vspace{0.4em}
\includegraphics[width=\linewidth]{figures/cluster3.pdf}

\caption{
Representative road-test scenarios sampled from the detected mismatch
regions. From top to bottom, the examples illustrate changes involving
large-vehicle cut-ins, temporary construction-related road conditions,
and U-turn interactions.
}
\label{fig:local_region_examples}
\end{figure}
The localized regions correspond to coherent and interpretable changes
in scenario composition. The first row of
Fig.~\ref{fig:local_region_examples} shows that cut-ins involving large
vehicles occur more frequently in Period~1. The second row captures
temporary construction-related changes in the testing route, which
induce consistent ego-vehicle deceleration followed by cut-ins from
surrounding agents. The third row contains U-turn interactions that
occur more frequently in Period~2. These examples support the
main-text observation that the proposed localization procedure captures
meaningful changes in road context and dynamic interaction patterns,
rather than isolated or semantically unrelated samples.

\subsubsection{Distributional Shift versus Metric Shift}
\label{app:distribution_metric_shift}
For the cross-week comparison, C2ST achieves an average classification
accuracy of $53.4\%$, corresponding to a classifier-detectable
discrepancy of
\[
2\widehat{\gamma}
=
2(0.534-0.5)
\approx 0.068.
\]
Over the same periods, the observed difference in the event metric is
\[
\left|
\widehat{\mu}_P-\widehat{\mu}_Q
\right|
=
|0.0624-0.0608|
\approx 0.0016.
\]

These quantities characterize different aspects of the distributional
change. The statistic $2\widehat{\gamma}$ measures the discrepancy
detectable from the complete latent scenario representation, whereas
the metric difference measures the effect of that change on one
specific downstream outcome. Thus, a detectable shift in the scenario
distribution need not induce an equally large change in every
performance metric.

\subsubsection{Full-scenes Experiments}

We further extend the proposed framework from cut-in scenarios to
full-scene road-testing data covering 13 representative interactive scenario categories. These include six ego-vehicle behaviors—lane changing, lane keeping, yielding, pulling out, pulling over, and nudging—and six traffic-participant behaviors—lane changing, parallel driving, cutting in, wrong-way driving, jaywalking, and backing up—together with a general category for scenarios that do not fall into any of the predefined classes. These categories are
automatically detected and annotated by the production data pipeline of
an industrial Level-4 AV system. To learn a unified full-scene latent
representation, we use complete road-testing logs collected over four
consecutive days. The training set is downsampled to balance the 13 scenario categories while preserving diverse road structures, traffic conditions, and multi-agent interaction patterns

For evaluation, we use full-day road-testing logs collected over three
additional consecutive days under a fixed AV system version, comprising approximately 90,000 scenarios per day. Fixing the
AV version reduces policy-related variation, allowing the comparison to
focus primarily on changes in scenario exposure and operational
conditions. We treat the first evaluation day as the source distribution and calibrate its metric estimate toward the distributions observed on the second and third days.  Table~\ref{tab:full_scene_calibration}
summarizes the global C2ST and metric-calibration results.

As shown in Table~\ref{tab:full_scene_calibration}, the event rate
increases from $0.011001$ on Day~1 to $0.011485$ and $0.011764$ on
Days~2 and~3, respectively. Directly transferring the Day~1 estimate
to the two target distributions therefore yields relative errors of
$4.21\%$ and $6.49\%$. The global C2ST results indicate moderate but
detectable cross-day distributional shifts, with a larger discrepancy
between Days~1 and~3. After localized importance calibration, the
estimated metrics move substantially closer to the corresponding target
reference values, reducing the remaining relative errors to
approximately $0.63\%$ and $0.97\%$, respectively. These preliminary
results demonstrate the applicability of the framework beyond
cut-in-only evaluation and can reduce distribution-induced metric bias in heterogeneous full-scene road-testing data.

\end{document}